\documentclass[runningheads]{llncs}

\makeatletter
\DeclareRobustCommand*\cal{\@fontswitch\relax\mathcal}
\makeatother

\usepackage[utf8]{inputenc}
\usepackage[T1]{fontenc}

\usepackage{amssymb}

\usepackage{amsmath}
\usepackage{stmaryrd}
\usepackage[ruled,vlined]{algorithm2e}
\usepackage{url}
\usepackage{paralist}
\usepackage[hidelinks]{hyperref}
\hypersetup{bookmarksdepth=3,bookmarksopen,bookmarksnumbered} 
\usepackage[numbers,sort&compress]{natbib}
\usepackage{todonotes}
\usepackage{datetime2}
\usepackage{xstring}
\usepackage{subcaption}
\usepackage{booktabs}
\usepackage{colortbl}
\usepackage{catchfile}
\usepackage[capitalize]{cleveref}

\usepackage{placeins}
\usepackage{mdframed}
\mdfsetup{skipabove=0pt,skipbelow=0pt,innertopmargin=0pt}
\usepackage[frozencache]{minted}
\definecolor{codehighlight}{RGB}{220, 255, 220}
\setminted{autogobble, breaklines, frame=none, framesep=3mm,  
           mathescape=true, escapeinside=||, highlightcolor=codehighlight,fontsize=\small} 
\definecolor{operator}{RGB}{0, 102, 187}
\newcommand{\op}[1]{\textcolor{operator}{#1}}
\usepackage{xspace}
\usepackage{bm}
\usepackage{pifont}

\usepackage{pgfplots}
\usepgfplotslibrary{groupplots}
\usepackage{pgfplotstable}

\newcommand{\T}{\mathcal{T}}

\newcommand{\A}{\mathcal{A}}
\newcommand{\sem}[1]{\llbracket#1\rrbracket}
\newcommand{\ignore}[1]{}
\newcommand{\angl}[1]{\langle#1\rangle}

\newcommand{\abs}[1]{\lvert#1\rvert}

\newcommand{\returnVar}{\textsf{ret}}
\newcommand{\jLj}{\textsf{jl\_joins}}
\newcommand{\lJluW}{\textsf{jl\_unlock\_after\_write}}
\newcommand{\luW}{\textsf{unlock\_after\_write}}
\newcommand{\lTlW}{\textsf{last\_tl\_write}}

\newcommand{\lTlL}{\textsf{last\_tl\_lock}}
\newcommand{\lTlU}{\textsf{last\_tl\_unlock}}

\newcommand{\lW}{\textsf{last\_write}}
\newcommand{\lV}{\textsf{last\_value}}

\newcommand\restr[2]{{\left.\kern-\nulldelimiterspace#1\vphantom{|}\right|_{#2}}}
\makeatletter
\newcommand*{\declarecommand}{%
  \@star@or@long\declare@command
}
\newcommand*{\declare@command}[1]{%
  \provide@command{#1}{}%
  \renew@command{#1}%
}
\makeatother

\newcommand{\single}{\textsf{single}}

\newcommand{\start}{u_0}
\newcommand{\act}{\textsf{act}}

\newcommand{\init}{\textsf{init}}
\newcommand{\Let}{\textbf{let}}
\newcommand{\Iif}{\textbf{if}}
\newcommand{\Then}{\textbf{then}}
\newcommand{\Eelse}{\textbf{else}}
\newcommand{\In}{\textbf{in}}

\newcommand{\Get}{\eta}

\newcommand{\return}{\textsf{return}}

\newcommand{\create}{\textsf{create}}
\newcommand{\new}{\textsf{new}}
\newcommand{\join}{\textsf{join}}
\newcommand{\lock}{\textsf{lock}}
\newcommand{\unlock}{\textsf{unlock}}

\newcommand{\self}{\textsf{self}}

\newcommand{\sink}{\textsf{sink}}
\newcommand{\last}{\textsf{last}}

\newcommand{\loc}{\textsf{loc}}
\newcommand{\id}{\textsf{id}}
\newcommand{\V}{{\cal V}}
\newcommand{\ValD}{{{\cal V}^\sharp}}
\newcommand{\N}{{\cal N}}
\newcommand{\E}{{\cal E}}

\newcommand{\X}{{\cal X}}

\newcommand{\R}{{\cal R}}
\declarecommand{\G}{{\cal G}}
\newcommand{\D}{{\cal D}}
\newcommand{\Vars}{{{\cal V}\!\mathit{ars}}}
\declarecommand{\C}{{\cal C}}
\newcommand{\I}{{\cal I}}
\newcommand{\TIDs}{\I}
\renewcommand{\SS}{{\cal S}}
\renewcommand{\C}{{\cal C}}
\newcommand{\M}{\textsf{M}}
\newcommand{\MM}{\mathcal{M}}
\newcommand{\GM}{\mathcal{G}}

\newcommand{\lift}{\textsf{lift}}
\newcommand{\unlift}{\textsf{unlift}}
\newcommand{\kClusters}{{\cal S}_k}
\makeatletter%
\@ifclassloaded{llncs}%
  {\newcommand{\True}{{\mathsf{True}}}}%
  {\newcommand{\True}{{\textsf{True}}}}%
\makeatother%
\makeatletter%
\@ifclassloaded{llncs}%
  {\newcommand{\False}{{\mathsf{False}}}}%
  {\newcommand{\False}{{\textsf{False}}}}%
\makeatother%

\makeatletter%
\@ifclassloaded{llncs}%
  {}%
  {}%
\makeatother%
\newcommand{\aSemR}[1]{\llbracket#1\rrbracket^\sharp_\R}
\newcommand{\aSem}[1]{\llbracket#1\rrbracket^\sharp}
\newcommand{\unique}{\textsf{unique}}
\newcommand{\lcommondefinite}{\textsf{lcu\_anc}}
\newcommand{\maycreate}{\textsf{may\_create}}
\newcommand{\Actions}{\mathcal{A}\mathit{ct}}
\newcommand{\semA}[1]{\aSem{#1}_\A}
\newcommand{\newA}{\new^\sharp_\A}
\newcommand{\accounted}{\textsf{acc}}
\newcommand{\initA}{\init^\sharp_\A}
\newcommand{\Clusters}{\mathcal{Q}}
\newcommand{\Cluster}{Q}

\begin{document}

\title{Clustered Relational Thread-Modular\\Abstract Interpretation with Local Traces}                     
\titlerunning{Clustered Relational Thread-Modular Abstract Interpretation}

 \author{
    Michael Schwarz\inst{1}\and
    Simmo Saan\inst{2}\and
    Helmut Seidl\inst{1}\and\\
    Julian Erhard\inst{1}\and
    Vesal Vojdani\inst{2}
}
\authorrunning{M. Schwarz et al.}
\institute{
     Technische Universit\"at M\"unchen, Garching, Germany\\
     \email{\{m.schwarz, helmut.seidl, julian.erhard\}@tum.de}\and
     University of Tartu, Tartu, Estonia\\
     \email{\{simmo.saan, vesal.vojdani\}@ut.ee}
}

\maketitle              
\begin{abstract}
    We construct novel thread-modular analyses that track relational information for potentially overlapping clusters of global variables 
    -- given that they are protected by common mutexes.
    We provide a framework to systematically increase the precision of clustered relational analyses by
    splitting control locations based on abstractions of \emph{local traces}.
    As one instance, we obtain an analysis of dynamic thread creation and joining. 
    Interestingly, tracking \emph{less relational} information for globals may result in higher precision.
    We consider the class of 2-decomposable domains that encompasses many weakly relational domains (e.g., \emph{Octagons}).
    For these domains, we prove that maximal precision is attained already for
    clusters of globals of sizes at most 2.
\end{abstract}

\section{Introduction}\label{s:intro}\label{s:relDomains}
    Tracking relationships between program variables is indispensable for proving properties of programs or
    verifying the absence of certain programming errors \cite{Cousot78,Cousot09,Logozzo08}. 
    %
    %
%
%
%
Inferring relational properties is particularly challenging for multi-threaded programs
as all interferences by other threads that may happen in parallel, must be taken into account.
    In such an environment, only relational properties between globals protected by common mutexes are likely to persist throughout program execution.
%
Generally, relations on clusters consisting of fewer variables are less brittle than those on larger clusters.
Moreover, monolithic relational analyses employing, e.g., the polyhedral abstract domain are known to be notoriously expensive~\cite{Mine01,Singh2017}.
Tracking \emph{smaller} clusters may even be more precise than tracking larger clusters
\cite{Farzan12}.
%

\begin{example}\label{e:cluster}
    Consider the following program. All accesses to globals $g$, $h$, and $i$ are protected by the mutex $a$.
    \begin{center}
        \begin{minipage}[t]{7cm}
    \begin{minted}{c}
    main:
     x = |\op{create}|(t1); y = |\op{create}|(t2);
     |\op{lock}|(a);
     g = ?; h = ?; i = ?;
     |\op{unlock}|(a); r = |\op{join}|(y); |\op{lock}|(a);
     z = ?; g = z; h = z; i = z;
     |\op{unlock}|(a); |\op{lock}|(a);
     // ASSERT(h==i); (1) ASSERT(g==h); (2)
     |\op{unlock}|(a);
    \end{minted}
        \end{minipage}
        \begin{minipage}[t]{2.4cm}
    \begin{minted}{c}
    t1:
     |\op{lock}|(a);
     x = h;
     i = x;
     |\op{unlock}|(a);
     |\op{return}| 1;
    \end{minted}
        \end{minipage}
        \begin{minipage}[t]{2.4cm}
            \begin{minted}{c}
            t2:
             |\op{lock}|(a);
             g = ?; h = ?;
             |\op{unlock}|(a);
             |\op{return}| 0;
            \end{minted}
        \end{minipage}
    \end{center}
    In this program, the main thread creates two new threads, starting at $t_1$ and $t_2$, respectively.
    Then it locks the mutex $a$
    to set all globals non-deterministically to some value and unlocks $a$ again. After having joined the thread $t_2$,
    it locks $a$ again and sets all globals to the \emph{same} unknown value and unlocks $a$ again.
    Thread $t_1$ sets $i$ to the value of $h$. Thread $t_2$ sets $g$ and $h$ to (potentially different) unknown values.
    Assume we are interested in equalities between globals.
    In order to succeed in showing assertion $(1)$, it is necessary to detect that the main thread is unique and thus cannot read its past writes since these have been overwritten.
    Additionally, the analysis needs to certify that thread $t_2$ also is unique, has been joined before the assertion, and that its writes must also have been overwritten.

    For an analysis to prove assertion $(2)$, propagating a joint abstraction of the values of all globals protected by $a$ does not suffice:
    At the unlock of $a$ in $t_1$, $g{=}h$ need not hold.
    If this monolithic relation 
    is propagated to the last lock of $a$ in main,  $(2)$ cannot be shown --- despite $t_1$ modifying neither $g$ nor $h$.\qed
\end{example}

Here we show, that the loss of precision indicated in the example can be remedied by
replacing the monolithic abstraction of all globals protected by a mutex with suitably chosen subclusters.
In the example, we propose to instead consider the subclusters $\{g,h\}$ and $\{h,i\}$ separately.
As $t_1$ does not write any values to the cluster $\{g,h\}$, the imprecise relation $\top$ thus is not propagated to the main thread and assertion $(2)$ can be shown.

To fine-tune the analysis, we rely on \emph{weakly} relational domains.
A variety of weakly relational domains have been proposed in the literature
such as \emph{Two Variables Per Inequality} \cite{Simon02}, \emph{Octagons} \cite{Mine01,Mine06Oct}, or simplifications thereof \cite{Mine01Zones,Logozzo08}.
%
The technical property of interest which all these domains have in common is that each abstract relation can be reconstructed from its projections onto subclusters of variables of size at most 2.
We call such domains 2-\emph{decomposable}.
Beyond the numerical $2$-decomposable domains, also non-numerical $2$-decomposable domains can be constructed
such as a domain relating string names and function pointers.

Based on $2$-decomposable domains,
we design thread-modular relational analyses of globals which may attain additional precision
by taking \emph{local} knowledge of threads into account.  
%
Therefore, we do not rely on a \emph{global} trace semantics, but on a
\emph{local} trace semantics which formalizes for each thread that part of the computational past it can observe~\cite{traces1}.
Abstract values for program points describe the set of all reaching local traces.
Likewise, values recorded for \emph{observable} actions
are abstractions of all local traces ending in the corresponding action.
Such observable actions are, e.g., unlock operations for mutexes.
The abstract values are then refined by taking finite
abstractions of local traces into account.
To this end, we propose a generic framework that re-uses the
components of any base analysis as black boxes.
%
%
%
Our contributions can be summarized as follows:
\begin{itemize}
	%
    \item We provide a new relational analysis of globals as an abstraction of the
	collecting local trace semantics based on overlapping clusters of variables (\cref{s:traces,s:base,s:clustering}).
    \item Our analysis deals with dynamically created and joined threads, whose thread \emph{id}s may, e.g., be communicated to other
     threads via variables and which may synchronize via mutexes (\cref{s:traces}).
    \item We provide a generic scheme to incorporate history-based arguments into the analysis by taking
	 finite abstractions of local traces into account (\cref{s:refinement}).
    \item We give an analysis of dynamically created thread \emph{id}s as an instance of our generic scheme.
    We apply this to exclude \emph{self-influences} or reads from threads that cannot possibly run in parallel (\cref{s:unique,s:self-excluded}).
    \item We prove that some loss of precision of relational analyses can be avoided by tracking \emph{all} subclusters of variables.
    For the class of 2-decomposable relational domains, 
	we prove that
	tracking variable clusters of size greater than $2$ can be abandoned without precision loss (\cref{s:clustering}).
\end{itemize}
    The analyses in this paper have all been implemented, a report of a practical evaluation
    is included in \cref{s:experiments}, whereas \cref{s:related} details related work.

\section{Relational Domains}\label{s:relational}
First, we define the notion of relational domain employed in the description of our analysis.
%
%
%
%
Let $\Vars$ be a set of variables, potentially of different types.
We assume all configurations and assignments to be well-typed,
i.e., the type of the (abstract) value matches the one specified for a variable.
For each type $\tau$ of values, we assume a complete lattice $\V_\tau^\sharp$ of abstract values
abstracting the respective concrete values from $\V_\tau$.
Let $\ValD$ denote the collection 
of these lattices, and $\Vars\to_\bot\ValD$ denote the set of all type-consistent assignments
$\sigma$ from variables to non-$\bot$ abstract values, extended with a dedicated least element (also denoted by $\bot$), and
equipped with the induced ordering.
%
A \emph{relational domain} $\R$ then is a complete lattice which provides the following operations\\
\noindent\begin{minipage}[t]{.65\linewidth}\vspace{-1.2em}\[
\begin{array}{rcl}
\aSemR{x\leftarrow e}&:&\R\to\R\text{ (assignment for expression $e$)}	\\
\restr{r}{Y} &:&\R\to\R\text{ (restriction to $Y \subseteq \Vars$)} \\
\aSemR{?e}&:&\R\to\R\text{ (guard for condition $e$)}
\end{array}
\]
  \end{minipage}%
  \begin{minipage}[t]{.35\linewidth}\vspace{-1.2em}\[
      \begin{array}{rcr}
      \lift&:&(\Vars\to_\bot\ValD)\to\R 	 \\
      \unlift	&:&\R\to(\Vars\to_\bot\ValD)
      \end{array}
      \]
\end{minipage}
\medskip

\noindent The operations to the left provide the abstract state transformers for the basic operation of programs (with non-deterministic assignments expressed as restrictions),
while $\lift$ and $\unlift$ allow casting from abstract variable assignments to the relational domain
as well as extracting single-variable information. We assume that
$\lift\,\bot = \bot$ and $\unlift\,\bot = \bot$, and require that
$\unlift\circ\lift \sqsupseteq \id$
where $\sqsupseteq$ refers to the ordering of $(\Vars\to_\bot\ValD)$.
Moreover, we require that the \emph{meet} operations $\sqcap$ of $\ValD$ and $\R$ safely approximate the intersection of the concretizations
of the respective arguments.
%
Restricting a relation $r$ to a subset $Y$ of variables amounts to \emph{forgetting} all information about variables not in $Y$.
Thus, we demand $\restr{r}{\Vars}=r$, $\restr{r}{\emptyset}=\top$, $\restr{r}{Y_1} \sqsupseteq \restr{r}{Y_2}$ when $Y_1 \subseteq Y_2$,
$\restr{(\restr{r}{Y_1})}{Y_2} = \restr{r}{Y_1 \cap Y_2}$, and
\begin{equation}
	\begin{array}{lr}
		\unlift\,(\restr{r}{Y})\,x = \top \quad (x \not\in Y) \qquad \qquad \qquad &
    \unlift\,(\restr{r}{Y})\,x = (\unlift\,r)\,x \quad (x \in Y)
	\end{array}
\end{equation}
Restriction thus is \emph{idempotent}.
For convenience, we also define a shorthand for assignment of abstract values\footnote{
We use $\sigma\oplus\{x_i\mapsto v_i\mid i=1,\ldots,m\}$ to denote the variable assignment
obtained from $\sigma$ by replacing the values for $x_i$ with $v_i$ ($i=1,\ldots,m$).}:
	$\aSemR{x \leftarrow^\sharp v}\,r = \left(\restr{r}{\Vars \setminus \{x\}}\right)\sqcap \left(\lift \left(\top\oplus\{ x \mapsto v \}\right)\right)$.
%
In order to construct an abstract interpretation,
we further require monotonic concretization functions
$\gamma_\ValD:\ValD\to 2^\V$ and $\gamma_\R:\R\to 2^{\Vars\to\V}$
satisfying the requirements presented in \cref{def:sound}.
\begin{figure}[b]
  \begin{mdframed}
\begin{equation*}
\begin{array}{l@{\;}l@{\;}l}
\forall a,b: a \sqsubseteq b \implies \gamma_\ValD\,a \subseteq \gamma_\ValD\, b
\qquad\;\; \gamma_\R\,\bot	= \emptyset \qquad\;  \forall r,s: r \sqsubseteq s \implies \gamma_\R\,r \subseteq \gamma_\R\, s \span\span \\
\gamma_\R\,(\aSemR{x \leftarrow e}\,r) &\supseteq& \{\sigma\oplus\{x\mapsto\sem{e}\sigma\}\mid \sigma\in\gamma_\R r\}	\\
\gamma_\R(\restr{r}{Y}) \supseteq
\{\sigma\oplus\{x_1\mapsto v_1,\ldots,x_m\mapsto v_m\}\mid
		v_i\in\V,
		x_i\in\Vars \setminus Y,
		\sigma\in\gamma_\R r\} \span\span	\\
\gamma_\R\,(\lift\,\sigma^\sharp)	\supseteq
	\{\sigma\mid\forall x
		:\,\sigma\,x\in\gamma_\ValD\,(\sigma^\sharp\,x)\}	\qquad\qquad\;\;
\gamma_\ValD\,(\unlift\,r)\,x	\supseteq
	\{\sigma\,x\mid\sigma\in\gamma_\R\,r\} \span\span
\end{array}
\end{equation*}
\end{mdframed}
\caption{Required properties for $\gamma_\ValD:\ValD\to 2^\V$ and $\gamma_\R:\R\to 2^{\Vars\to\V}$.}\label{def:sound}
\end{figure}
%
\begin{example}\label{e:rel}
As a value domain $\V^\sharp_\tau$, consider the flat lattice over the sets of values of appropriate type $\tau$.
A relational domain $\R_1$ is obtained by collecting satisfiable conjunctions of equalities between variables or variables and constants
where the ordering is logical implication, extended with $\False$ as least element.
The greatest element in this complete lattice is given by $\True$.
The operations $\lift$ and $\unlift$ for non-$\bot$ arguments then can be defined as
\[
\begin{array}{lllllll}
\lift\,\sigma	= \bigwedge\{x=\sigma\,x\mid x\in\Vars, \sigma\,x\neq\top\}	 \qquad\,\,
\unlift\,r\,x =  \begin{cases}
		c	& \text{ if } r\implies(x=c)	\\
		\top	& \text{otherwise}
		 \end{cases}
\end{array}
\]
The restriction of $r$ to a subset $Y$ of variables is given by the conjunction
of all equalities implied by $r$ which only contain variables from $Y$ or constants.
\qed
\end{example}

\noindent
In line of \cref{e:rel}, relational domains for non-numerical values may also be constructed.


A variable clustering ${\cal S}\subseteq 2^{\Vars}$ is a set of subsets (\emph{clusters}) of variables.
For any cluster $Y\subseteq\Vars$, let $\R^{Y} = \{r \mid r\in\R, \restr{r}{Y} = r \}$;
this set collects all abstract values from $\R$
containing information on variables in $Y$ only.
Given an arbitrary clustering ${\cal S}\subseteq 2^{\Vars}$,
\emph{any} relation $r\in\R$ can be approximated by a meet of
relations from $\R^Y$ ($Y\in\cal S$) since for every $r\in\R$,
$r\sqsubseteq \mbox{$\bigsqcap$}\{\restr{r}{Y}\mid Y\in{\cal S}\}$
holds.

Some relational domains, however, can be fully recovered from their restrictions
to specific subsets of clusters. We consider for $k\geq 1$,
the set $\kClusters$ of all non-empty subsets $Y\subseteq\Vars$ of cardinality at most $k$.
We call a relational domain $\R$ \emph{$k$-decomposable} if
each abstract value from $\R$ can be precisely expressed as the meet of its restrictions to clusters of
$\kClusters$ and when all least upper bounds can be recovered by computing with clusters of $\kClusters$ only; that is,
\begin{eqnarray}
  \begin{array}{lll}
  r =& \mbox{$\bigsqcap$} \left\{ \restr{r}{\Cluster} \mid Q\in\kClusters \right\}	\qquad\quad
  \restr{\left(\mbox{$\bigsqcup$} R\right)}{Q} =&\mbox{$\bigsqcup$}\left\{\restr{r}{Q}\mid r\in R\right\}\quad(Q\in \kClusters)
	\label{def:decomp}
  \end{array}
\end{eqnarray}
holds for each abstract relation $r\in\R$ and each set of abstract relations $R\subseteq\R$.
%

\begin{example}\label{e:domains}

%
%
%

The domain $\R_1$ from the previous example is $2$-decomposable.
This also holds for the
\emph{octagon} domain \cite{Mine01} and many other weakly relational numeric domains
(pentagons~\cite{Logozzo08},
weighted hexagons~\cite{Fulara10},
logahedra~\cite{HoweK09},
TVPI~\cite{Simon02},
dDBM~\cite{Peron07}, and
AVO~\cite{Chen14}).
The \emph{affine equalities} or \emph{affine inequalities} domains \cite{Karr1976,Cousot78}, however,
are not.
The relational string domains recently proposed by \citet{Arceri22} in Sec.\ 5.1 through Sec.\ 5.3, are examples of non-numeric $2$-decomposable domains.
%
\end{example}

%

\section{A Local Trace Semantics}\label{s:traces}

We build upon the semantic framework for \emph{local traces}, introduced by~\citet{traces1}.
A local trace records all past events that have affected the present configuration of a specific thread,
referred to as the \emph{ego} thread.
In~\cite{traces1}, the local trace semantics is proven equivalent to the global trace semantics
which itself is equivalent to a global interleaving semantics.
In particular, any analysis that is sound w.r.t.\ the local trace semantics also is w.r.t.\ the interleaving semantics.

While the framework of \citet{traces1} allows for different formalizations of traces,
thread synchronization happens only via locking/unlocking and thread creation.
Generalizing their semantics, we identify certain actions as
\emph{observable} by other threads when executing corresponding \emph{observing} actions (see \cref{fig:paradigms}).
When the \emph{ego} thread executes an \emph{observing} action, a local trace ending in the corresponding \emph{observable} action is incorporated.
%
%
Here, we consider as observable/observing actions
locking/unlocking mutexes and creating/joining threads.

\begin{table}[b]
    \vspace{-1.5em}
    \caption{Observable and observing actions and which concurrency primitive they relate to.
    The primitives targeted by this paper are in bold font.}\label{fig:paradigms}
    \centering
    \begin{tabular}{ccc}
        \toprule
        Observable Action & Observing Action & Programming Concept   \\
        \midrule
        $\unlock(a)$ & $\lock(a)$ & \textbf{Mutex}, Monitor, ... \\
        $\return\,x$ & $x' {=} \join(x'')$ & \textbf{Thread Returning / Joining} \\
        $g=x$ & $x=g$ & \textbf{Writing/Reading a global variable}\\
        $\texttt{signal(c)}$ & $\texttt{wait(c)}$ & Condition Variables \\
        $\texttt{send(chan,v)}$ & $\texttt{x = receive(chan)}$ & Channel-Based Concurrency, Sockets, ...\\
        $\texttt{set\_value}$ & $\texttt{get}$ & Futures / Promises \\
        \bottomrule
    \end{tabular}
\end{table}

Consider, e.g., the program in \cref{f:ex0} and a corresponding local trace (\cref{f:trace0}).
This trace consists of one \emph{swim lane} for each thread
representing the sequence of steps it executed where each node in the graph
represents a configuration attained by it.
Additionally, the trace records the \emph{create} and \emph{join} orders as well as
for each mutex $a$, the \emph{locking} order for $a$
($\to_c,\to_j$, and $\to_a$, respectively).
These orders introduce extra relationships between thread configurations.
%
The unique start node of each local trace is an initial configuration of the \emph{main}
thread.
%

We distinguish between the sets $\X$ and $\G$ of \emph{local} and \emph{global} variables.
We assume that $\X$ contains a special variable $\self$ within which the thread \emph{id} of the current thread,
drawn from the set $\I$, is maintained.
A (local) \emph{thread configuration} is a pair $(u,\sigma)$ where $u$ is a program point and the type-consistent map
$\sigma: \X \to \V$ provides values for the local variables.
The values of globals are \emph{not} explicitly represented in a thread configuration,
but can be recovered by consulting the (unique) last write to this global within the
local trace. To model weak memory effects, weaker notions of last writes are conceivable.
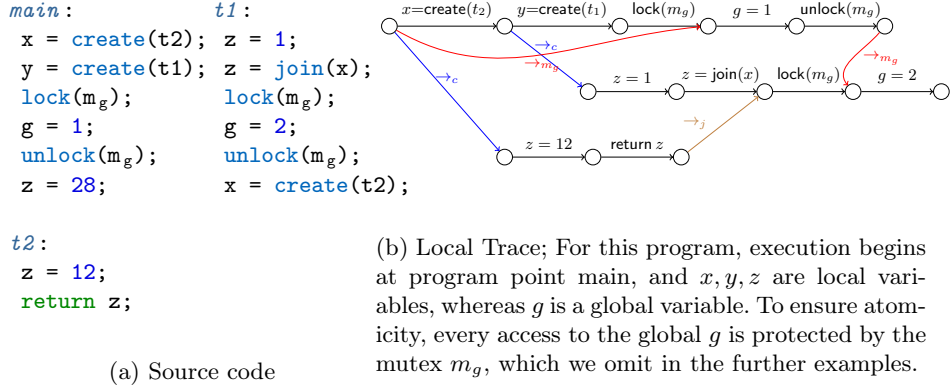
\begin{figure}[t]
    \begin{subfigure}{0.4\textwidth}
\begin{center}
    \begin{minipage}[t]{2.7cm}
    \begin{minted}{c}
    main:
     x = |\op{create}|(t2);
     y = |\op{create}|(t1);
     |\op{lock}|(m|$_{\,\texttt{g}}$|);
     g = 1;
     |\op{unlock}|(m|$_{\,\texttt{g}}$|);
     z = 28;

    t2:
     z = 12;
     return z;
    \end{minted}
    \end{minipage}\begin{minipage}[t]{2.7cm}
    \begin{minted}{c}
    t1:
     z = 1;
     z = |\op{join}|(x);
     |\op{lock}|(m|$_{\,\texttt{g}}$|);
     g = 2;
     |\op{unlock}|(m|$_{\,\texttt{g}}$|);
     x = |\op{create}|(t2);
    \end{minted}
    \end{minipage}
\end{center}
\subcaption{Source code}\label{f:ex0}
\end{subfigure}\begin{subfigure}{0.6\textwidth}
        \usetikzlibrary{calc}
        \usetikzlibrary{fit,shapes.geometric}
        \tikzset{
            every node/.style={node distance=30pt and 40pt},
            programpoint/.style={circle,draw,font=\small},
            edgelabel/.style={midway,font=\small,above=0.5mm}}
        \scalebox{0.68}{
        \begin{tikzpicture}
            \node[programpoint,node distance=40pt and 55pt](ppm1){};
            \node[programpoint,node distance=40pt and 55pt,right=of ppm1](pphalb){};
            \node[programpoint,node distance=40pt and 55pt,right=of pphalb](pp1){};
            \node[programpoint,right=of pp1](pp2){};
            \node[programpoint,right=of pp2](pp3){};
            \node[programpoint,right=of pp3](pp4){};

            \node[programpoint,below right=of pphalb](ppa){};
            \node[programpoint,right=of ppa](ppb){};
            \node[programpoint,right=of ppb](ppc){};
            \node[programpoint,right=of ppc](ppd){};
            \node[programpoint,right=of ppd](ppe){};

            \node[programpoint,below left=of ppa](ppo1){};
            \node[programpoint,right=of ppo1](ppo2){};
            \node[programpoint,right=of ppo2](ppo4){};

            \draw[->](pphalb)--(pp1)node[edgelabel]{$y{=}\create(t_1)$};
            \draw[->](ppm1)--(pphalb)node[edgelabel]{$x{=}\create(t_2)$};
            \draw[->](pp1)--(pp2)node[edgelabel]{$\lock(m_g)$};
            \draw[->](pp2)--(pp3)node[edgelabel]{$g=1$};
            \draw[->](pp3)--(pp4)node[edgelabel]{$\unlock(m_g)$};

            \draw[->](ppo1)--(ppo2)node[edgelabel]{$z=12$};
            \draw[->](ppo2)--(ppo4)node[edgelabel]{$\return\,z$};

            \draw[->,blue](pphalb)--(ppa)node[edgelabel]{$\to_c$};
            \draw[->,blue](ppm1)--(ppo1)node[edgelabel]{$\to_c$};
            \draw[->,brown](ppo4) to node[edgelabel,left=2mm]{$\to_j$} (ppc);

            \draw[->](ppa)--(ppb)node[edgelabel]{$z=1$};
            \draw[->](ppb)--(ppc)node[edgelabel]{$z = \join(x)$};
            \draw[->](ppc)--(ppd)node[edgelabel]{$\lock(m_g)$};
            \draw[->](ppd)--(ppe)node[edgelabel]{$g=2$};

            \draw[->,red, out=-30,in=185](ppm1) to node[edgelabel, below=0mm]{$\to_{m_g}$} (pp2);
            \draw[->,red, out=230,in=130](pp4) to node[edgelabel, right=1mm]{$\to_{m_g}$} (ppd);

        \end{tikzpicture}}
        \vspace{1em}
        \subcaption{Local Trace; For this program, execution begins at program point main, and
        $x,y,z$ are local variables, whereas $g$ is a global variable. To ensure atomicity, every access to the
        global $g$ is protected by the mutex $m_g$, which we omit in the further examples.}\label{f:trace0}
    \end{subfigure}

\caption{An example program and a corresponding local trace.}
\end{figure}
As in \cite{traces1}, we consider a set of actions $\Actions$ that consists of locking
and unlocking a (non-reentrant) mutex from a set $\M$, copying values of globals into locals and vice-versa,
creating a new thread, as well as assignments with and branching on local variables.
We extend $\Actions$ with actions for \emph{returning} from and \emph{joining} with threads.
We assume that writes to and reads from globals are \emph{atomic} (or more precisely,
we assume copying values of \emph{integral} type 
to be atomic).
This is enforced for each
global $g$ by a dedicated mutex $m_g$ acquired just before accessing $g$ and released immediately after.
For simplicity, we associate traces corresponding to a write of $g$ to this dedicated mutex $m_g$, and
thus do not need to consider writing and reading of globals as observable/observing actions.
In examples, we omit explicitly locking and unlocking these mutexes.
By convention, at program start all globals have value $0$, while local
variables may initially have any value.

Each thread is represented by a control-flow graph with edges $e\in\E$ of the form $e=(u,\act,u')$ for some action $\act\in\Actions$ and program points
$u$ and $u'$ where the start point of the \emph{main} thread is $\start$.
Let $\T$ denote the set of all local traces of a given program.
A formalism for local traces must, for each edge $e$ of the control-flow graph, provide a transformation
$\sem{e}:\T^k\to 2^\T$ so that $\sem{e}(t_0,\ldots,t_{k-1})$ extends the local trace $t_0$,
possibly incorporating other local traces.
For the operations $\lock(a),a\in\M$, or $x{=}\join(x'), x,x'\in\X$, the arity of $\sem{e}$ is two,
another local trace, namely, with last operation $\unlock(a)$ or $\return\,x''$, respectively, is incorporated.
The remaining edge transformations have arity one.
In all cases, the set of resulting local traces may be empty when the operation is not
applicable to its argument(s).
We write $\sem{e}(T_0,\ldots,T_{k-1})$ for the set
$\bigcup_{t_0\in T_0,\ldots,t_{k-1}\in T_{k-1}}\sem{e}(t_0,\ldots,t_{k-1})$.

\begin{figure}[b]
    \begin{mdframed}
    \begin{minipage}[t]{.435\linewidth}\[
        \begin{array}{lll}
            \sem{u,\lock(a)}\,\Get = (\emptyset,\sem{e}(\Get\,[u],\Get\,[a]))\\[1ex]
            \sem{u,\unlock(a)}\,\Get
                =\\
            \;\;\Let\;T = \sem{e}(\Get\,[u])\;\In\;\\
            \;\;(\{[a]\mapsto T\}, T)	\\[1ex]
            \sem{u,x=g}\,\Get	= (\emptyset,\sem{e}(\Get\,[u]))	\\[1ex]
            \sem{u,g=x}\,\Get	= (\emptyset,\sem{e}(\Get\,[u]))
        \end{array}
    \]
    \end{minipage}%
    \begin{minipage}[t]{.58\linewidth}\[
        \begin{array}{lll}
        \sem{u,x {=} \create(u_1)}\,\Get = \Let\;T = \sem{e}(\Get\,[u])\;\In	\\
        \;\;		(\{[u_1]\mapsto\new\,u\,u_1\,(\Get\,[u])\},T)	\\[1.2ex]

        \sem{u,x {=}\join(x')}\;\Get = \Let\;T = \Get\,[u]\;\In\\
        \;\; (\emptyset,\sem{e}(\Get\,[u],\bigcup\{
                        \Get\,[t\,(x')]\mid t\in\Get\,[u]\}))
            \\[1.2ex]
        %
        %
        \sem{u,\return\,x}\;\Get
            =
        \Let\;T = \Get\,[u]\;\In\\
        \;\; (\{[i]\mapsto\sem{e}(\{t\in T\mid t(\self) = i\})\mid i\in\I\},\sem{e}T)
        \end{array}
        \]
    \end{minipage}
    \end{mdframed}
    \caption{Right-hand sides for side-effecting formulation of concrete semantics; $t(y)$ extracts the value of local variable $y$
    from the terminal configuration of trace $t$.
    }\label{f:concreteSide}
\end{figure}

Given definitions of $\sem{e}$, the set $\T$ can be inductively defined starting from a set $\init$ of
\emph{initial local traces} consisting of initial configurations of the \emph{main} thread.
To develop efficient thread-modular abstractions,
we are interested in subsets $\T[u], \T[a],\T[i]$ of local traces ending at some program point $u$, ending with an \emph{unlock} operation for mutexes $a$ (or from $\init$), or ending with a \emph{return} statement of thread $i$, respectively.
%
\citet{traces1} showed that 
such subsets can be described as the least solution of a 
\emph{side-effecting constraint system}~\cite{apinis2012side}.
There, each
right-hand side may, besides its contribution to the unknown on the left, also provide contributions to other unknowns (the \emph{side-effects}).
This allows expressing analyses that accumulate flow-insensitive information about globals during a flow-sensitive analysis of local states with dynamic control flow~\cite{td3lang}.
Here, in the presence of dynamic thread creation,
we use side-effects to express that an observable action, unlock or return,
should \emph{also} contribute to the sets $\T[a]$ or $\T[i]$, such that they can be incorporated
at the corresponding observing action.
The side-effecting formulation of our concrete semantics takes the form:
    \begin{equation}
	\begin{array}{lllllll}
        (\eta,\eta\,[u_0])	\sqsupseteq	(\{ [a] \mapsto \init \mid a {\in}\M \}, \init)\quad
        (\eta,\eta\,[u'])	\sqsupseteq	\sem{u,\act}\eta \;\; (u,\act,u'){\in}\E
	\end{array}
	\label{c:concrete}
    \end{equation}
where the ordering $\sqsupseteq$ is induced by the superset ordering and right-hand sides are defined in~\cref{f:concreteSide}.
A right-hand side takes an assignment $\eta$ of the unknowns of the system
and returns a pair $(\eta',T)$ where
$T$ is the contribution to the unknown occurring on the left (as in ordinary constraint systems).
The first component collects the side-effects as the assignment $\eta'$.
If the right-hand sides are monotonic, \cref{c:concrete} has a unique least solution. 

%
We only detail the right-hand sides for the \emph{creation} of threads as well as
the new actions \emph{join} and \emph{return};
the rest remain the same as defined by~\citet{traces1}.
%
For \textbf{thread creation}, they provide the action $x {=} \create(u_1)$.
Here, $u_1$ is the program point at which the created thread should start.
We assume that all locals from the creator are passed to the created thread, except for the variable $\self$.
The variables $\self$ in the created thread and $x$ in the creating thread receive a fresh thread \emph{id}.
Here,
$\new\,u\,u_1\,t$ computes the local trace at the start point $u_1$ from the local trace $t$ of the creating thread.
To handle \textbf{returning} and \textbf{joining} of threads we introduce the following two actions:
\begin{itemize}
	\item	\return\ $x$; -- terminating a thread and returning the value of the local variable $x$
		to a thread waiting for the given thread to terminate.
	\item $x {=} \join(x');$ where $x'$ is a local variable holding a thread \emph{id}
		-- blocks the ego thread, until the thread with
		the given thread \emph{id} has terminated. As in \emph{pthreads}, at most one
		thread may call join for a given thread \emph{id}. The value provided to \return\ by the
		joined thread is assigned to the local variable $x$.
\end{itemize}
For returning results and realization of \textsf{join},
we employ the unknown $[i]$ for the thread \emph{id} $i$ of the returning thread, as shown in \cref{f:concreteSide}.

\section{Relational Analyses as Abstractions of Local Traces}\label{s:base}\label{S:BASE}
Subsequently, we give relational analyses of the values of globals which we base on the local trace semantics.
They are generic in the relational domain $\R$, with $2$-decomposable domains being particularly well-suited,
as the concept of \emph{clusters} is central to the analyses.
We focus on relations between globals that are jointly \emph{write}-protected by some mutex.
%
We assume we are given for each global $g$, a set $\MM[g]$  of (write) \emph{protecting} mutexes, i.e.,
mutexes that are \emph{always} held when $g$ is written.
%
Let $\GM[a]=\{g\in\G\mid a\in\MM[g]\}$ denote the set of globals protected by a mutex $a$.
Let $\emptyset \neq \Clusters_a \subseteq 2^{\GM[a]}$ the set of clusters of these globals we associate with $a$.
For technical reasons, we require at least one cluster per mutex $a$, which may be the empty cluster $\emptyset$,
thus not associating any information with $a$.

Our basic idea is to store at the unknown $[a,\Cluster]$ (for each mutex $a$ and cluster $\Cluster\in\Clusters_a$)
an abstraction of the relations \emph{only} between globals in $\Cluster$. By construction, all globals in $\Cluster$ are protected by $a$.
Whenever it is locked, the relational information stored at all $[a,\Cluster]$ is incorporated into the local state
by the lattice operation \emph{meet}, i.e., the local state now maintains relations between locals \emph{as well as} globals which no
other thread can access at this program point.
Whenever $a$ is unlocked, the new relation between globals in all corresponding clusters
$\Cluster\in\Clusters_a$ is side-effected to the respective unknowns $[a,\Cluster]$.
Simultaneously, all information on globals no longer protected, is \emph{forgotten} to obtain the new local state.
In this way, the analysis is fully relational in the local state, while only keeping relations within clusters of globals jointly protected
by some mutex.

For clarity of presentation,
we perform \emph{control-point splitting} on the set of held mutexes when reaching program points. Apart from this, the constraint system
and right-hand sides for the analysis closely follow those of the concrete semantics (\cref{s:traces}) ---
with the
exception that unknowns now take values from $\R$ and that unknowns $[a]$ are
replaced with unknowns $[a,\Cluster]$ for $\Cluster\in\Clusters_a$.
\begin{figure}[t]
	\begin{mdframed}
    \begin{minipage}[t]{.5\linewidth}\[
        \begin{array}{lll}
			\init^\sharp\,\Get =\\
			\;\;\Let\;r(\Cluster) = \aSemR{\{ g \leftarrow 0 \mid g\in\Cluster\}} \top\;\In\\
			\;\;\Let\;\rho = \left\{ [a,\Cluster] \mapsto r (\Cluster) \mid a \in \M, \Cluster\in\Clusters_a  \right\}\;\\
			\;\;\In\;(\rho,\aSemR{\self \leftarrow^\sharp i_{0}}\top)\\[1ex]

			\sem{[u,S],x {=} \create(u_1)}^\sharp\Get	=\\
				\;\;\Let\;r = \Get\,[u,S]\;\In\\
				\;\;\Let\; i = \nu^\sharp\,u\,u_1\,r\;\In\\
				\;\;\Let\; r' = \restr{\left\{\aSemR{\self \leftarrow^\sharp i}\,r\right\}}{\X}\;\In\\
				\;\; \Let\; \rho = \{ [u_1,\emptyset] \mapsto r' \}\;\In\\
				\;\;	(\rho,\aSemR{x \leftarrow^\sharp i}r)	\\[1ex]

			\sem{[u,S],g=x}^\sharp\Get =\\
				\;\;(\emptyset, \aSemR{g \leftarrow x}\,(\Get\,[u, S]))\\[1ex]

			\sem{[u,S],x=g}^\sharp\Get =\\
				\;\; (\emptyset, \aSemR{x \leftarrow g}\,(\Get\,[u, S]))

        \end{array}
    \]
    \end{minipage}%
    \begin{minipage}[t]{.55\linewidth}\[
        \begin{array}{lll}
			\sem{[u,S],\lock(a)}^\sharp\Get	=\\
			\;\; \left(\emptyset,\Get\,[u, S]\sqcap\left(\bigsqcap_{\Cluster\in\Clusters_a}\Get\,[a,\Cluster]\right)\right)\\[1ex]

			\sem{[u,S],\unlock(a)}^\sharp\Get =\\
			\;\; \Let\;r = \Get\,[u, S]\;\In	\\
			\;\;	\Let\;\rho = \{ [a,\Cluster] \mapsto \restr{r}{\Cluster} \mid \Cluster\in\Clusters_a\}\;\In \\
			\;\; \left(\rho,\restr{r}{\X \cup \bigcup \{\GM[a'] \mid a' \in (S\setminus{a})\}}\right)\\[1ex]

			\sem{[u,S],\return\,x}^\sharp\Get =\\
			\;\; \Let\;r = \Get\,[u,S]\;\In\\
			\;\; \Let\;i^\sharp = \unlift\,r\,\self\;\In\\
			\;\; \left(\left\{ [i^\sharp] \mapsto \restr{\left(\aSemR{\returnVar \leftarrow x}\,r\right)}{\{\returnVar\}}\right\},r\right)\\[1ex]

			\sem{[u,S],x' {=} \join(x)}^\sharp\Get =\\
			\;\; \Let\;v = \bigsqcup_{\unlift\,r\,x'\sqcap i'  \neq \bot} \unlift\,(\Get[i'])\,\returnVar  \;\In\\
			\;\;\left(\emptyset,\aSemR{x' \leftarrow^\sharp v}(\Get\,[u,S])\right)
        \end{array}
        \]
    \end{minipage}
	\end{mdframed}
    \caption{Right-hand sides for the basic analysis.
	All functions are strict in $\bot$
	(describing the empty set of local traces), we only display definitions for non-$\bot$ abstract values here.
	$\aSemR{\{ g \leftarrow 0 \mid g\in\Cluster\}}$ is shorthand for the abstract transformer corresponding to the assignment of $0$
	to all variables in $\Cluster$ one-by-one.
	}\label{f:basicAbs}
\end{figure}

All right-hand sides are given in detail in \cref{f:basicAbs}. 
For the \textbf{start point} of the program and the empty lockset, the right-hand side $\init^\sharp$
returns the $\top$ relation updated such that the variable $\self$ holds
the abstract thread \emph{id} $i_{0}$ of the \emph{main} thread.
Additionally, $\init^\sharp$
produces a side-effect for each mutex $a$ and cluster $\Cluster$
that initializes all globals from the cluster with the value $0$.

For a \textbf{thread creating} edge starting in program point $u$ with lockset $S$, the right-hand side $\sem{[u,S],x {=} \create(u_1)}^\sharp$ first generates a new abstract thread \emph{id},
which we assume can be computed using the function $\nu^\sharp$.
The new \emph{id} is assigned to the variable $x$ in the local state of the current thread.
Additionally, the start state $r'$ for the newly created thread is constructed
and side-effected to the thread's start point with empty lockset $[u_1,\emptyset]$.
Since threads start with empty lockset, the state $r'$ is obtained by removing all information
about globals from the local state of the creator
and assigning the new abstract thread \emph{id} to the variable $\self$.
%

When \textbf{locking} a mutex $a$, the states stored at unknowns $[a,\Cluster]$ with $\Cluster\in\Clusters_a$ are combined with the local state by \emph{meet}.
This is sound because the value stored at any $[a,\Cluster]$ only maintains relationships between variables write-protected by $a$, and these values soundly account for the program state at every $\unlock(a)$ and at program start.
When \textbf{unlocking} $a$, on the other hand,
the local state restricted to the appropriate clusters $\Cluster\in\Clusters_a$
is side-effected to the respective unknowns $[a,\Cluster]$, so that the changes made to variables in the cluster become visible to other threads.
Also, the local state is restricted to the local variables and only those globals for which at least one
protecting mutex is still held.

As special mutexes $m_g$ immediately surrounding accesses to $g$ are used to ensure atomicity, and information about $g$ is associated with them,
all \textbf{reads} and \textbf{writes} refer to the local copy of $g$. \textbf{Guards} and \textbf{assignments} (which may only involve local variables) are defined analogously.

For a \textbf{return} edge,
the abstract value to be returned is looked up in the local state and then side-effected to the abstract thread \emph{id} of the current thread
(as the value of the dedicated variable $\returnVar$).
For a \textbf{join} edge, the least upper bound of all \emph{return} values of all possibly joined threads
is assigned to the left-hand side of the \emph{join} statement in the local state.
\begin{example}
	Consider the program\footnote{In all examples, 
	$g$, $h$, and $i$ are globals, whereas $x$, $y$, and $z$ are locals, and the clusters at special mutexes $m_g$
	contain exactly the global $g$:
	$\Clusters_{m_g} = \{\{g\}\}$. Unless explicitly stated otherwise, the domain $\R_1$ from \cref{e:rel}, enhanced with inequalities between variables is used.}
	where $\MM[g] = \{a,b,m_g\}$, $\MM[h] = \{a,b,m_h\}$, $\Clusters_a = \{\{g,h\}\}$,
	$\Clusters_b = \{\{g,h\}\}$.
	\begin{center}
		\begin{minipage}[t]{4cm}
\begin{minted}{c}
main:
 x = |\op{create}|(t1); y = ?;
 |\op{lock}|(a); |\op{lock}|(b);
 g = y; h = y+9;
 |\op{unlock}|(b); |\op{lock}|(b);
 h = y;
 // ASSERT(g==y); (1)
 // ASSERT(h==y); (2)
 |\op{unlock}|(b); |\op{unlock}|(a);
 x = |\op{create}|(t2);
\end{minted}
		\end{minipage}
		\begin{minipage}[t]{3.8cm}
\begin{minted}{c}
t1:
 |\op{lock}|(b);
 |\op{unlock}|(b);
 |\op{lock}|(a);
 |\op{lock}|(b);
 // ASSERT(g==h); (3)
 y = ?; g = y; h = y;
 |\op{unlock}|(b);
 |\op{unlock}|(a);
\end{minted}
\end{minipage}
\begin{minipage}[t]{4cm}
\begin{minted}{c}
t2:
 |\op{lock}|(b);
 |\op{lock}|(a);
 // ASSERT(g==h); (4)
 |\op{unlock}|(a);
 |\op{unlock}|(b);

\end{minted}
\end{minipage}
\end{center}
%
Our analysis succeeds in proving all assertions here. Thread $t_2$ is of particular interest:
When locking $b$ only $g\leq h$ is known to hold, and locking the additional mutex $a$ means that the better
information $g=h$ becomes available.
The analysis by \citet{Mukherjee2017} on the other hand only succeeds in proving assertion (2) --- even when all
globals are put in the same region.
It cannot establish (1) because all
correlations between locals and globals are forgotten when the \emph{mix} operation is applied at the second $\lock(b)$
in the main thread. (3) cannot be established because, at $\lock(b)$ in $t_1$, the mix operation also incorporates
the state after the first $\unlock(b)$ in the main thread, where $g=h$ does not hold. Similarly, for (4).
The same applies for assertion (3) and the analysis using \emph{lock invariants} proposed by \citet{Mine2014}.
This analysis also falls short of showing (1), as at the $\lock(b)$ in the main thread, the lock
invariant associated with $b$ is joined into the local state. (4) is similarly out of reach.
The same reasoning also applies to
\cite{Mukherjee2017,Mine2014,traces1}
after equipping the analyses
with thread \emph{id}s.\qed
%
%
\end{example}

\begin{theorem}\label{t:sound}
Any solution of the constraint system is sound w.r.t.\ the local trace semantics.
\end{theorem}
\begin{proof}
	The proof is by fixpoint induction, the details are given in \cref{s:baseSound}.
\end{proof}
\noindent
We remark that, instead of relying on $\MM[g]$ being pre-com\-puted, an analysis can also infer this information on the fly~\cite{Vojdani2016}.
\textsc{Goblint}, e.g.,
is capable of dealing with the resulting non-monotonicity in the side-effects.

The analysis described here however still has some deficiencies.
All writes to a global are accumulated regardless of the writing thread. As a consequence,
a thread does, e.g., not only read its latest local writes but also all earlier local writes,
even if those are \emph{definitely} overwritten.
Excluding some threads' writes is an instance of a more general idea for excluding
writes that cannot be last writes.
Instead of discussing ad hoc ways to address this specific shortcoming, we propose a general mechanism to improve the precision of any thread-modular analysis in the next section,
and later instantiate it to the issue highlighted here.

\section{Refinement via Finite Abstractions of Local Traces}\label{s:refinement}\label{S:REFINEMENT}
To improve precision of thread-modular analyses
we take additional abstractions of local traces into account. Our approach is generic in that it builds
on the right-hand sides of a base analysis and uses them as black boxes.
We will later instantiate this framework to exclude writes based on thread \emph{id}s from the analysis in \cref{s:base}.
Other instantiations are conceivable as well.
To make it widely applicable, the framework allows base analyses that already perform some splitting
of unknowns at program points (e.g., locksets as seen in \cref{s:base}). We denote by $[\hat{u}]$ such (possibly)
extended unknowns for a program point $u$.

A (base) analysis is defined by its right-hand sides, and a collection of different domains:
(1) $\D_S$ for abstract values stored at unknowns for program points;
(2) $\D_\act$ for abstract values stored at observable actions $\act$ (e.g., in \cref{s:base}, $\D_M$ for unlocks and $\D_T$ for thread returns).

Let $\A$ be some set of finite information that can be extracted from a local
trace by a function $\alpha_\A: \T \to \A$. We call $\alpha_\A\,t\in\A$ the \emph{digest} of some local trace $t$.
Let $\semA{u,\act}: \A^k\to2^\A$ be the effect on the digest when performing a
$k$-ary action $\act\in\Actions$
for a control flow edge originating at $u$.
We require that for $e=(u,\act,v)\in\E$,
\begin{equation}
    \begin{array}{lll}
    \forall A_0,\ldots,A_{k-1} \in \A&:&\; \abs{\semA{u,\act}(A_0,\ldots,A_{k-1})}\leq 1	\\
    \forall t_0,\ldots,t_{k-1} \in \T&:&
    \alpha_\A(\sem{e}(t_0,\ldots,t_{k-1})) \subseteq \semA{u,\act}(\alpha_\A\,t_0,\ldots,\alpha_\A\,t_{k-1})
    \end{array}
	\label{def:Asound}
\end{equation}
where $\alpha_\A$ is lifted element-wise to sets.
While the first restriction ensures determinism,
the second intuitively ensures that $\aSem{u,\act}_\A$ soundly abstracts $\sem{e}$.

For thread creation, we additionally require a helper function $\newA: \N\to\N\to\A\to\A$ that returns for a thread created
at an edge originating from $u$ and starting execution at program point $u_1$ the new digest.
The same requirements are imposed for edges $(u,x {=}\create(u_1),v)\in\E$,
\begin{equation}
    \begin{array}{llllll}
        \forall A_0 {\in} \A&:& \abs{\newA\,u\,u_1\,A_0}\leq 1	\quad\;
        \forall t_0 {\in} \T&:&
        \alpha_\A(\new\,u\,u_1\,t) \subseteq \newA\,u\,u_1\,(\alpha_\A\,t_0)
    \end{array}
    \label{def:AnewSound}
\end{equation}
\noindent Also, we define for the initial digest at the start of the program
\begin{equation}
    \begin{array}{lll}
        \initA &=& \{\alpha_A\,t \mid t \in\init\}
    \end{array}
	\label{def:AinitSound}
\end{equation}
Under these assumptions, we can perform
\emph{control-point splitting} according to $\A$.
%
%
This means that unknowns $[\hat{u}]$ for program points $u$ are replaced with new unknowns
$[\hat{u},A]$, $A\in\A$.
Analogously, unknowns for observable actions $[\act]$ are replaced with unknowns $[\act,A]$ for $A\in\A$.
%
%
Consider a single constraint from an abstract constraint system of the last section, which soundly abstracts the collecting local
trace semantics of a program.
\[
\begin{array}{llllllrrr}
(\Get,\Get\,[\hat{v}])	&\sqsupseteq& \aSem{[\hat{u}],\act}\,\Get	\quad\,\,
\end{array}
\]
The corresponding constraints of the refined system with control-point splitting
differ based on whether
the action $\act$ is observing, observable, or neither.
\begin{itemize}
    \item When $\act$ is \emph{observing}, the new right-hand side additionally
    gets the digest $A_1$ associated with the local traces that are to be incorporated:
        \[
        \begin{array}{lll}
            (\Get, \Get\,\left[\hat{v}, A'\right]) &\sqsupseteq&
                    \aSem{[\hat{u},A_0],\act,A_1}\,\Get
            \qquad \text{for } A_0,A_1\in\A, A'\in\semA{u,\act}\,(A_0,A_1)
        \end{array}
        \]
    \item When $\act$ is $\emph{observable}$, the digest $A'$
    of the resulting local trace is passed to the new right-hand side, such that the side-effect can be redirected to the appropriate unknown:
    \[
        \begin{array}{lll}
            (\Get, \Get\,\left[\hat{v}, A'\right]) &\sqsupseteq&
                    \aSem{[\hat{u},A_0],\act,A'}\,\Get
            \qquad \text{for } A_0\in\A, A'\in\semA{u,\act}\,(A_0)
        \end{array}
    \]
    \item When $\act$ is neither, 
        no additional digest is passed:
    \[
        \begin{array}{lll}
            (\Get, \Get\,\left[\hat{v}, A'\right]) &\sqsupseteq&
                    \aSem{[\hat{u},A_0],\act}\,\Get
            \qquad\;\;\;\;\; \text{for } A_0\in\A, A'\in\semA{u,\act}\,(A_0)
        \end{array}
    \]
\end{itemize}
\begin{figure}[t]
\begin{mdframed}
    \begin{minipage}[t]{.4\linewidth}\[
        \begin{array}{lll}
            \aSem{[\hat{u},A_0],\act,A_1}\,\Get =\\
            \;\;\Let\,\Get'\,[x] = \Iif\;[x] = [\hat{u}]\;\Then\\
            \quad\;\;\;\Get\,[\hat{u},A_0]\\
            \quad\Eelse\,\Get\,[x,A_1]\;\\
            \;\;\In\\
            \;\;\aSem{[\hat{u}],\act}\,\Get'\\[1ex]

            \aSem{[\hat{u},A_0],\act''}\,\Get =\\
            \;\;\Let\,\Get'\,[x] = \Get\,[x,A_0]\;\In \\
            \;\;\aSem{[\hat{u}],\act''}\,\Get'
        \end{array}
    \]
    \end{minipage}%
    \begin{minipage}[t]{.52\linewidth}\[
        \begin{array}{lll}
            \aSem{[\hat{u},A_0],\act',A'}\,\Get =\\
            \;\;\Let\,\Get'\,[x] = \Get\,[x,A_0]\;\In \\
            \;\;\Let\, (\rho,v) = \aSem{[\hat{u}],\act'}\,\Get'\;\In\\
            \;\;\Let\, \rho' = \left\{ [x,A'] \mapsto v' \mid ([x] \mapsto v') \in\rho \right\}\;\In\\
            \;\;(\rho',v)\\[1ex]

            \aSem{[\hat{u},A_0],x{=}\create(u_1)}\,\Get =\\
            \;\;\Let\,\Get'\,[x] = \Get\,[x,A_0]\;\In \\
            \;\;\Let\, (\{[\hat{u_1}] \mapsto v'\},v) = \aSem{[\hat{u}],x{=}\create(u_1)}\,\Get'\;\In\\
            \;\;(\{ [\hat{u_1},A'] \mapsto v' \mid A'\in \newA\,u\,u_1\,A_0\},v)
        \end{array}
        \]
    \end{minipage}
    \end{mdframed}
    \caption{Right-hand sides for an observing action $\act$, an observable action $\act'$, a create action, and an action $\act''$ that is neither
    for the refined analyses,
    defined as wrappers around the right-hand sides of a base analysis.}\label{f:refinedRHS}

\end{figure}
The new right-hand sides are defined in terms of the right-hand side of the base analysis which
are used as black boxes (\cref{f:refinedRHS}).
They act as wrappers,
mapping any unknown consulted or side-effected to by the original analysis
to the appropriate unknown of the refined system. 
Thus, the refined analysis automatically benefits from the extra information the digests provide. It may, e.g., exploit
that $\semA{u,\act}(A_0,A_1) = \emptyset$ meaning
that, no local traces with digests $A_0,A_1$ can be combined into a valid local trace ending with action $\act$.

The complete definition of the refined constraint system instantiated to the actions
considered in this paper and unknowns for program points enriched with locksets
can be found in the additional material
(\cref{f:abstractRef} in \cref{s:refineComplete}).

\begin{figure}[b]
    \begin{mdframed}
    \begin{minipage}[t]{.52\linewidth}\[
        \begin{array}{lll}
		\initA = \{\emptyset\}\\[0.5ex]
        \newA\,u\,u_1\,S = \{\emptyset\}\\[0.5ex]
		\semA{u,a}\,S = \{S\} \quad \text{(other non-observing)}\qquad\;\;
        \end{array}%
    \]
    \end{minipage}%
    \begin{minipage}[t]{.4\linewidth}\[
        \begin{array}{lll}
            \semA{u,\lock(a)}\,(S,S') = \{S \cup \{a\}\}\\[0.5ex]
			\semA{u,\unlock(a)}\,S = \{S \setminus \{a\}\}\\[0.5ex]
            \semA{u,a}\,(S,S') = \{S\} \quad \text{(other observing)}
        \end{array}
        \]
    \end{minipage}
	\end{mdframed}
    \caption{Right-hand sides for expressing locksets as a refinement.}\label{f:locksetA}
\end{figure}
    Enriching program points with locksets can in fact be seen as a first application of this framework.
    The right-hand sides are given in \cref{f:locksetA}.

\begin{example}
    As a further instance, consider tracking which mutexes have been locked at least once in the local trace.
    At $\lock(a)$ traces in which a thread has performed a $\lock(a)$ can not be combined with traces that contain no lock action for $a$.
    The corresponding right-hand sides are given in \cref{f:lockcountA}.

    When refining the analysis from \cref{s:base} accordingly (assuming that $a$ protects $g$ and $h$), it succeeds in proving the assert in this program as the initial values of $0$ for $g$ and $h$ can be excluded.
    \begin{center}
		\begin{minipage}[t]{4cm}
\begin{minted}{c}
main:
 |\op{lock}|(a);
 h = 9; g = 10;
 |\op{unlock}|(a);
 x = |\op{create}|(t1);
\end{minted}
		\end{minipage}
		\begin{minipage}[t]{3.8cm}
\begin{minted}{c}
t1:
 x = |\op{create}|(t2);
 |\op{lock}|(a);
 h = 11; g = 12;
 |\op{unlock}|(a);
\end{minted}
\end{minipage}\begin{minipage}[t]{3.8cm}
    \begin{minted}{c}
    t2:
     |\op{lock}|(a);
     // ASSERT(h<=g);
     |\op{unlock}|(a);
    \end{minted}
    \end{minipage}
\end{center}
    This could be naturally generalized to obtain instances counting how often some action (e.g., a write to a global $g$) occurred,
    stopping exact bookkeeping at some constant ($1$ in this example).
    \qed
\end{example}
\begin{figure}[b]
	\begin{mdframed}
    \begin{minipage}[t]{.35\linewidth}\[
        \begin{array}{lll}
		\initA = \{\emptyset\}\\[0.5ex]
        \newA\,u\,u_1\,L = \{L\}\\[0.5ex]
        \semA{u,a}\,S = \{L\} \;\; \text{(other non-observing)}\qquad\;\;\;
        \end{array}
    \]
    \end{minipage}%
    \begin{minipage}[t]{.5\linewidth}\[
        \begin{array}{lll}
            \semA{u,\lock(a)}\,(L,L') = \begin{cases}
                \emptyset & \text{if } a {\in} L \land a {\not\in} L'\\
                \{ L {\cup} L' {\cup} \{a\} \} & \text{otherwise}
            \end{cases}\\
            \qquad\qquad\quad\semA{u,a}\,(L,L') = \{L {\cup} L'\} \;\; \text{(other observing)}
        \end{array}
        \]
    \end{minipage}
	\end{mdframed}
    \caption{Right-hand sides for refining according to encountered $\lock$ operations.}\label{f:lockcountA}
\end{figure}

%
To prove soundness of local-trace-based refinement of our analysis from
\cref{s:base}, we first construct a corresponding refined collecting local trace semantics.
Then we verify that the refined analysis is sound w.r.t.\ this refined semantics --
which, in turn, is proven sound w.r.t.\ the original collecting local trace semantics.
%
\begin{theorem}\label{t:refined}
    Assume that
    $\alpha_\A$, $\newA$, and $\semA{u,\act}$ fulfill requirements \eqref{def:Asound}, \eqref{def:AnewSound}, and \eqref{def:AinitSound}.
    Then any solution of the refined constraint system is sound relative to the collecting local trace semantics.
\end{theorem}
\begin{proof}
    A proof sketch instantiated with the actions considered in this work
    and unknowns at program points enriched with locksets
    is provided in \cref{s:refineSound}.
\end{proof}
%
%
%

\section{Analysis of Thread Ids and Uniqueness}\label{s:unique}
\begin{figure}[b]
	\vspace{-0.5em}
		\begin{minipage}[t]{6cm}
\begin{minted}[tabsize=2]{c}
main:
 x = g; // PP u1
 y = |\op{create}|(t1);
 for(i = 0; i < 5; i++) { // PP u2
	 z = |\op{create}|(t1); }
\end{minted}
		\end{minipage}\begin{minipage}[t]{4cm}
	\begin{minted}[tabsize=2]{c}
t1:
 g = 42; // PP u3
 y = |\op{create}|(t1);

	\end{minted}
		\end{minipage}
\vspace{-0.5em}
\caption{Program with multiple thread creations.}\label{f:tid}
\end{figure}
\newcommand{\main}{\textsf{main}}
\noindent We instantiate the scheme from the previous section to compute abstract thread \emph{id}s and their uniqueness.
That refinement of the base analysis
enhances precision of the analysis by excluding reads, e.g., from threads that have not yet been started.
%
For that, we identify threads by their thread creation history, i.e., by sequences of \emph{create} edges.
As these sequences may grow arbitrarily, we collect all creates
occurring after the first repetition into a \emph{set} to obtain finite abstractions.
%
\begin{example}\label{e:tid1}
In the program from \cref{f:tid}, the first thread created by \emph{main} receives the abstract thread \emph{id}
$(\main\cdot\angl{u_1,t_1},\emptyset)$. It creates a thread with abstract thread \emph{id}
$(\main\cdot\angl{u_1,t_1}\cdot\angl{u_3,t_1},\emptyset)$.
At program point $u_3$, the latter creates a thread starting at $t_1$ and receiving the abstract thread \emph{id}
$(\main\cdot\angl{u_1,t_1},\{\angl{u_3,t_1}\})$ -- as do all threads subsequently created at this edge.
\qed
\end{example}
\noindent
%
Create edges, however, may also be repeatedly encountered within the creating thread, in a loop.
To deal with this, we track for each thread, the set $C$ of possibly already encountered \emph{create} edges.
As soon as a \emph{create} edge is encountered again, the created thread receives a non-unique thread \emph{id}.

\begin{example}\label{e:tid2}
The first time the main thread reaches program point $u_2$ in the program from \cref{f:tid},
the created
thread is assigned the \emph{unique} abstract thread \emph{id} $(\main\cdot\angl{u_2,t_1},\emptyset)$.
In subsequent
loop iterations, the created threads are no longer kept separate, and thus receive the non-unique \emph{id}
$(\main,\{\angl{u_2,t_1}\})$.
\qed
\end{example}

\noindent
Formally, let $\N_C,\N_S$ denote the subsets of program points with outgoing edge labeled $x {=} \create(...)$,
and of
starting points of threads, respectively.
Let $\cal P \subseteq \N_C \times \N_S$ denote sets of pairs relating thread creation nodes with the starting points of the created threads.
%
The set $\TIDs^\sharp$ of abstract thread \emph{id}s then consists of all
pairs $(i,s)\in (\main \cdot {\cal P}^*)\times2^{\cal P}$ in which each pair $\angl{u,f}$ occurs at most once.
Given the set $\TIDs^\sharp$, we require that there is a concretization $\gamma:\TIDs^\sharp\to2^\I$ and a
function $\single:\TIDs^\sharp\to\V^\sharp_\I$ with
$\gamma\,i^\sharp\subseteq\gamma_\ValD\,(\single\,i^\sharp)$.
The abstract thread \emph{id} of the \emph{main} thread
is given by $(\main,\emptyset)$.
Therein, the elements in $(\main \cdot {\cal P}^*)\times\{\emptyset\}$ represent the \emph{unique} thread \emph{id}s
representing at most one concrete thread \emph{id},
while the elements $(i,s)$, $s\neq\emptyset$, are \emph{ambiguous}, i.e., may represent multiple concrete thread \emph{id}s.
Moreover, we maintain the understanding that the concretizations of distinct abstract thread \emph{id}s from $\I^\sharp$ all are disjoint.
%

As refining information $\A$
we consider not only abstract thread \emph{id}s  --
but additionally track sets of executed thread creations within the current thread.
Accordingly, we set
$\A = \TIDs^\sharp \times 2^P$
%
and define the right-hand sides as seen in \cref{f:threadA}, where $\bar i$ denotes the set of pairs occurring in the sequence $i$.
\begin{figure}[b]
	\begin{mdframed}[innertopmargin=0pt]
	\begin{minipage}[t]{.6\linewidth}\[
	\begin{array}{lll}
		\initA = \{((\main,\emptyset),\emptyset)\}\\[0.5ex]
		\semA{u,x{=}\create(u_1)}\,(i,C) = \{(i, C \cup \{ \angl{u,u_1} \})\}\\[0.5ex]

		\semA{u,a}\,(i,C) = \{(i, C)\} \qquad \text{(for other actions $a$)} \\[0.5ex]

		\newA\,u\,u_1\,((d,s),C) =\\
			\quad \Let\;(d',s') = (d,s)\circ\angl{u,u_1}\;\In\\
			\quad \Iif\;s'=\emptyset \land \angl{u,u_1}\in C\;\Then\; ((d,\{\angl{u,u_1}\}), \emptyset)\\
			\quad \Eelse\;((d',s'), \emptyset)
        \end{array}
    \]
    \end{minipage}%
    \begin{minipage}[t]{.4\linewidth}\[
        \begin{array}{lll}
			(d,s)\circ\angl{u,u_1} =\\
			\quad\Iif\;d = (d_0\cdot\angl{u,u_1})\cdot d_1\;\Then\;\\
			\qquad(d_0,s\cup\bar{d_1}\cup\{\angl{u,u_1}\})	\\
			\quad\Eelse\;\Iif\;s=\emptyset\;\Then\;(d\cdot\angl{u,u_1},\emptyset)	\\
			\quad\Eelse\;(d,s\cup\{\angl{u,u_1}\})
        \end{array}
        \]
    \end{minipage}
	\end{mdframed}
    \caption{Right-hand sides for thread \emph{id}s.}\label{f:threadA}
\end{figure}
\begin{example}
Consider again the program from \cref{f:tid} with right-hand sides from \cref{f:threadA},
and assume that the missing right-hand for join returns its first argument.
The initial thread has the abstract thread \emph{id} $i_0 = (\main,\emptyset)$.
At its start point, the digest thus is $(i_0,\emptyset)$.
At the create edge originating at $u_1$, a new thread with \emph{id} $(\main\cdot\angl{u_1,t_1},\emptyset)$ is created.
The digest for this thread then is $((\main\cdot\angl{u_1,t_1},\emptyset),\emptyset)$.
For the main thread, the encountered create edge $\angl{u_1,t_1}$ is added to the second component of the digest, making
it $(i_0,\{\angl{u_1,t_1}\})$.

When $u_2$ is reached with $(i_0,\{\angl{u_1,t_1}\})$, a unique thread with \emph{id} $(\main\cdot\angl{u_2,t_1},\emptyset)$ is created.
The new digest of the creating thread then is $(i_0,\{\angl{u_1,t_1},$ $\angl{u_2,t_1}\})$.
In subsequent iterations of the loop, for which $u_2$ is reached with $(i_0,\{\angl{u_1,t_1}, \angl{u_2,t_1}\})$,
a non-unique thread with \emph{id}
$(\main,\{\angl{u_2,t_1}\})$ is created.

When reaching $u_3$ with \emph{id} $(\main,\{\angl{u_2,t_1}\})$, a thread with \emph{id}
$(\main,\{\angl{u_2,t_1},$ $\angl{u_3,t_1}\})$ is created as the \emph{id} of the creating thread was already not unique.
When reaching it with the \emph{id} $(\main\cdot\angl{u_1,t_1},\emptyset)$, a new thread with \emph{id}
$(\main\cdot\angl{u_1,t_1}\cdot\angl{u_3,t_1},\emptyset)$ is created.
When the newly created thread reaches this program point, the threads created there have the \emph{non-unique} \emph{id}
$(\main\cdot\angl{u_1,t_1},\{\angl{u_3,t_1}\})$, as $\angl{u_3,t_1}$ already
appears in the \emph{id} of the creating thread.\qed
\end{example}
\noindent
Abstract thread \emph{id}s should provide us with functions
\begin{itemize}
	\item $\unique: \TIDs^\sharp {\to} \textbf{bool}$ tells whether a thread \emph{id} is unique.
	\item $\lcommondefinite: \TIDs^\sharp {\to} \TIDs^\sharp {\to} \TIDs^\sharp$ returns the last common \emph{unique} ancestor  of two threads.
	\item $\maycreate: \TIDs^\sharp {\to} \TIDs^\sharp {\to} \textbf{bool}$ checks whether a thread \emph{may} (transitively) create another.
\end{itemize}
For our domain $\TIDs^\sharp$, these can be defined as $\unique\,(i,s) = (s = \emptyset)$ and
\[
\begin{array}{lll}
\lcommondefinite\,(i,s)\,(i',s') &=& (\textsf{longest common prefix } i\,i', \emptyset)\\
\maycreate\,(i,s)\,(i',s') &=& (\bar i \cup s) \subseteq(\bar{i'} \cup s')
\end{array}
\]
%
%
We use this extra information to enhance the definitions of $\semA{u,\lock(a)}$
and $\semA{u,x' {=} \join(x)}$ to take into account that
the ego thread cannot acquire a mutex from another thread or join a thread
that has definitely not yet been created.
This is the case for a thread $t'$
\begin{itemize}
	\item[(1)] 
	that is directly created by the \emph{unique} ego thread, but the ego thread has 
	not yet reached the program point where $t'$ is created;
	\item[(2)] whose thread \emph{id} indicates that a thread that has not yet been created according to (1), is part
	of the creation history of $t'$.
\end{itemize}
Accordingly, we introduce the predicate $\textsf{may\_run}\,(i,C)\,(i',C')$ defined as
\[
\begin{array}{lll}
(\lcommondefinite\,i\,i' = i) \implies

\exists \angl{u,u'} \in C: (i {\circ} \angl{u,u'} = i' \lor \maycreate\,(i {\circ} \angl{u,u'})\,i')
\end{array}
\]
which is false whenever thread $i'$ is definitely not yet started. We then set
\[
\begin{array}{lll}
	\semA{u,\lock(a)}\,(i,C)\,(i',C') &=& \semA{u,x' {=} \join(x)}\,(i,C)\,(i',C')\\
	 &=& \begin{cases}
		\{(i,C)\} & \text{if } \textsf{may\_run}\,(i,C)\,(i',C')\\
		\emptyset & \text{otherwise}
	\end{cases}
\end{array}
\]
This analysis of thread \emph{id}s and uniqueness
can be considered as a \emph{May-Happen-In-Parallel} (or, more precisely, \emph{Must-Not-Happen-In-Parallel}) analysis.
MHP information is useful in a variety of scenarios:
a thread-modular analysis of \emph{data races} or \emph{deadlocks}, e.g.,
that does not consider thread \emph{id}s and joining yet,
can be refined by means of the given analysis of thread \emph{id}s
to exclude some data races or deadlocks that the original analysis can not.
Subsequently, we outline how the analysis from \cref{s:base} may benefit from MHP information.

\section{Exploiting Thread \emph{ID}s to Improve Relational Analyses}\label{s:not-yet-started}\label{s:self-excluded}\label{S:SELF-EXCLUDED}
We subsequently exploit
abstract thread \emph{id}s and their uniqueness to
limit the amount of reading performed by the analysis from \cref{s:base}.
\begin{itemize}
	\item[\textbf{I1}] from other threads that have not yet been created.
	\item[\textbf{I2}] the ego thread's past writes, if its thread \emph{id} is unique.
	\item[\textbf{I3}] past writes from threads that have already been joined.
\end{itemize}
Improvements \textbf{I1} and \textbf{I3} have, e.g., been realized in a setting where thread \emph{id}s and which thread is joined where
can be read off from control-flow graphs \cite{Kusano16}.
Here, however, this information is computed \emph{during} analysis.
%
In our framework,
\textbf{I1} is already achieved by refining the base analysis
according to \cref{s:unique}.
%
\begin{example}\label{e:shortcomings}
	Consider the program below where $\MM[g] = \{a,b,m_g\}$, $\MM[h] = \{a,b,m_h\}$, $\MM[i] = \{m_i\}$
	and assume $\Clusters_a =\{\{g,h\}\}$.
	\begin{center}
		\begin{minipage}[t]{6.5cm}
\begin{minted}{c}
main:
 x = |\op{create}|(t1); |\op{lock}|(a);
 // ASSERT(g==h);  (1)
 |\op{unlock}|(a);
 y = |\op{create}|(t2); |\op{lock}|(a);
 // ASSERT(g==h);  (2)
 g = 42; h = 42;
 |\op{unlock}|(a); z = |\op{create}|(t3);
 i = 3;  i = 2; // ASSERT(i==2); (3)
 i = 8;
\end{minted}
		\end{minipage}
		\begin{minipage}[t]{5cm}
\begin{minted}{c}
t1:
 |\op{lock}|(a);
 r = ?; g = r; h = r;
 |\op{unlock}|(a);

t2:
 |\op{lock}|(a); v = g; |\op{unlock}|(a);

t3:
 |\op{lock}|(a); g = 19; |\op{unlock}|(a);
\end{minted}
\end{minipage}
\end{center}
The analysis succeeds in proving assertion (1),
as the thread (starting at) $t_3$ that breaks
the invariant $g{=}h$ has definitely not been started yet
at this program point. Without refinement, the analysis from \cref{s:base} could not prove (1).
However, this alone does not suffice to prove (2).
At this program point, thread $t_2$ may already be started.
At the $\lock(a)$ in $t_2$, thread $t_3$ may also be started;
thus, the violation of the invariant $g{=}h$ by $t_3$ is incorporated
into the local state of $t_2$ at the lock.
At $\unlock(a)$, despite $t_2$ \emph{only reading} $g$, the imprecise abstract relation where $g{=}h$
does not hold, is side-effected to $[a,\{g,h\},t_2]$ and is
incorporated at the second $\lock(a)$ of the main thread.
The final shortcoming is that each thread reads all its own past (and future!) writes --
even when it is known to be unique.
This means that (3) cannot be proven.
\qed
\end{example}
\noindent
To achieve \textbf{I2}, some effort is required as our analysis forgets values of globals when they become unprotected.
This is in contrast, e.g., to \cite{Mine2014,Mukherjee2017}.
We thus restrict side-effecting to mutexes
to cases where the ego thread has possibly written a protected global
since acquiring it.
This is in contrast to \cref{s:base}, where a side-effect is performed at every unlock, i.e., everything a thread reads is
treated as if it was written by that thread.

Technically, we locally track a
map $L: (\M \times \Clusters) \to \R$, where $L\,(a,\Cluster)$ maintains for a mutex $a$,
an abstract relation between the globals in cluster $\Cluster \in \Clusters_a$.
More specifically, the abstract relation
on the globals from $\Cluster$ recorded in $L\,(a,\Cluster)$ is the one that held when $a$ was unlocked \emph{join-locally}
for the \emph{first} time after the \emph{last} \emph{join-local} write to a global in $\GM\,[a]$.
If there is no such $\unlock(a)$, the relation at program start is recorded.
We call an operation in a local trace \emph{join-local} to the ego thread,
if it is (a) \emph{thread-local}, i.e., performed by the ego thread, or (b) is executed by a thread that is (transitively) joined into the ego thread, or (c) is \emph{join-local} to
the parent thread at the node at which the ego thread is created.
This notion will also be crucial for realizing \textbf{I3}.
\emph{Join-locality} is illustrated in \cref{f:joinLocal}, where the \emph{join-local} part of a local trace is highlighted.

For \emph{join-local} contributions, by construction, it suffices to consult $L\,a$ instead of unknowns $[a,\Cluster,i]$.
Such contributions are called \emph{accounted} for.
%
To check whether a contribution from a specific thread \emph{id} is \emph{accounted} for, we introduce the
function
$\accounted: (\A \times \D_S) {\to} \A {\to} \textsf{bool}$
(see definition \eqref{e:accounted} below).
Besides an abstract value from $\R$, the local state $\D_S$ now
contains two additional components:
\begin{itemize}
	\item The map $L: (\M \times \Clusters) \to \R$ for which the join is given component-wise;
	\item The set $W: 2^\G$ (ordered by $\subseteq$) of globals that may have been written since one of its protecting mutexes has
	been locked, and not all protecting mutexes have been unlocked since.
\end{itemize}
Just like $r$, $L$ and $W$ are abstractions of the reaching local traces.
$\D_T$ is also enhanced with an $L$ component, 
while $\D_M$ remains unmodified.
\begin{figure}[b]
    \usetikzlibrary{calc}
    \usetikzlibrary{fit,shapes.geometric}
    \centering
    \tikzset{
        every node/.style={node distance=25pt and 40pt},
        programpoint/.style={circle,draw,font=\small},
        edgelabel/.style={midway,font=\small,above=0.5mm}}
    \scalebox{0.75}{
    \begin{tikzpicture}
        \node[programpoint,node distance=20pt and 55pt](ppm1){};
        \node[programpoint,node distance=20pt and 55pt,right=of ppm1](pphalb){};
        \node[programpoint,node distance=20pt and 55pt,right=of pphalb](pp1){};
        \node[programpoint,right=of pp1](pp2){};
        \node[programpoint,right=of pp2](pp3){};
        \node[programpoint,right=of pp3](pp4){};

        \node[programpoint,below right=of pphalb](ppa){};
        \node[programpoint,right=of ppa](ppb){};
        \node[programpoint,right=of ppb](ppc){};
        \node[programpoint,right=of ppc](ppd){};
        \node[programpoint,right=of ppd](ppe){};

        \node[programpoint,below left=of ppa](ppo1){};
        \node[programpoint,right=of ppo1](ppo2){};
        \node[programpoint,right=of ppo2](ppo4){};

        \draw[->](pphalb)--(pp1)node[edgelabel]{$y{=}\create(t_1)$};
        \draw[->](ppm1)--(pphalb)node[edgelabel]{$x{=}\create(t_2)$};
        \draw[->](pp1)--(pp2)node[edgelabel]{$\lock(m_g)$};
        \draw[->](pp2)--(pp3)node[edgelabel]{$g=1$};
        \draw[->](pp3)--(pp4)node[edgelabel]{$\unlock(m_g)$};

        \draw[->](ppo1)--(ppo2)node[edgelabel]{$z=12$};
        \draw[->](ppo2)--(ppo4)node[edgelabel]{$\return\,z$};

        \draw[->,blue](pphalb)--(ppa)node[edgelabel]{$\to_c$};
        \draw[->,blue](ppm1)--(ppo1)node[edgelabel]{$\to_c$};
        \draw[->,brown](ppo4) to node[edgelabel,left=2mm]{$\to_j$} (ppc);

        \draw[->](ppa)--(ppb)node[edgelabel]{$z=1$};
        \draw[->](ppb)--(ppc)node[edgelabel]{$z {=} \join(x)$};
        \draw[->](ppc)--(ppd)node[edgelabel]{$\lock(m_g)$};
        \draw[->](ppd)--(ppe)node[edgelabel]{$g=2$};

        \draw[->,red, out=-30,in=185](ppm1) to node[edgelabel, below=0mm]{$\to_{m_g}$} (pp2);
        \draw[->,red, out=230,in=130](pp4) to node[edgelabel, right=1mm]{$\to_{m_g}$} (ppd);

    \draw[draw=green, rounded corners, line width=8, fill=green!80!black, opacity=0.1]
    (ppm1.north west) -- (pphalb.north west) -- (ppa.north east) -- (ppe.north east) |-
    (ppo4.south west) |- (ppo1.south west) -- (ppm1.south west) |- (ppm1.north west);
    \end{tikzpicture}}
    \caption{Illustration highlighting the \emph{join-local} part of a local trace of the program from \cref{f:ex0},
    and which writes are thus accounted for by $L$.}

    \label{f:joinLocal}
\end{figure}
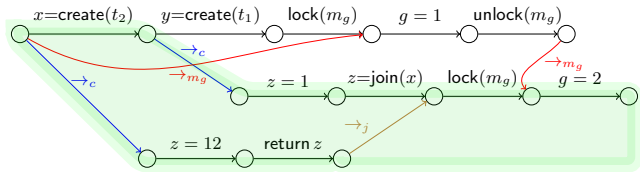
We sketch the right-hand sides here, definitions are given in \cref{f:improvedAbs}.
For \textbf{program start} $\init^\sharp$,
in contrast to the analysis from \cref{s:base}, there is no initial side-effect to the unknowns for mutexes.
The initial values of globals are \emph{join-local}, and thus accounted for in the $L$ component also passed to any subsequently created thread.

The right-hand sides for \textbf{thread creation} and \textbf{return} differ from the analysis
from \cref{s:base} enhanced with thread \emph{id}s only in the handling of additional data structures $L$ and $W$.
%
As the thread \emph{id}s are tracked precisely in the $\A$ component, this information is directly used when
determining which unknown to side-effect to and unknowns $[(i,C)]$ replace unknowns $[i',(i,C)]$.

For a \textbf{join} edge,
if the return value of the thread is not accounted for, it is assigned to the variable on the
left-hand side and the $L$ information from the ego thread and the joined thread is joined together.
If, on the other hand, it is accounted for, then the thread has already been joined and cannot be joined here again.
There is a separate constraint for each $(i',C')$, so that all threads that could be joined here are considered.

For \textbf{locking} of mutexes, upon lock, if $(i',C')$ is not accounted for, its information on the globals protected by
$a$ is joined with the \emph{join-local} information for $a$ maintained in $L\,(a,\Cluster)$, $\Cluster\in\Clusters_a$.
This information about the globals protected by $a$ is then incorporated into the local state by $\sqcap$.
For \textbf{unlocking} of mutexes, if there may have been a write to a protected global since the mutex was locked
(according to $W$), the \emph{join-local} information is updated
and the local state restricted to $\Cluster$ is
side-effected to the appropriate unknown $[a,\Cluster,(i,C)]$ for $\Cluster\in\Clusters_a$.
Just like in \cref{s:base}, $r$ is then restricted to only maintain relationships between locals and those
globals for which at least one protecting mutex is still held.
\textbf{Reading} from and \textbf{writing} to globals once more are purely local operations.
To exclude self writes, we set
\begin{equation}
\begin{array}{lll}
\accounted\,((i,C),\_)\,(i',C') = \unique\,i \land i = i'
\end{array}\label{e:accounted}
\end{equation}
The resulting analysis thus takes {\bf I1} (via $\semA{...}$ defined in \cref{s:unique}), as well as {\bf I2}
(via $\accounted$) into account.
In \cref{e:shortcomings}, it is now able to show all assertions.

\begin{figure}[t]
    \begin{mdframed}\begin{minipage}[t]{.53\linewidth}\[
        \begin{array}{lll}
            \init^\sharp_{(i,C)} =\\
            \;\;\Let\;L\,(a,\Cluster) = \aSemR{ \{g \leftarrow 0 \mid g\in\Cluster\}} \top\;\In \\
            \;\;\Let\;r = \aSemR{\self \leftarrow^\sharp i}\top\;\In\\
            \;\;(\emptyset,(\left\{ (a,\Cluster) \mapsto L\,(a,\Cluster) \mid a{\in}\M, \Cluster{\in}\Clusters_{a}\right\},\emptyset,r))\\[1ex]

            \sem{[u,S,(i,C)],x' = \join(x),(i',C')}^\sharp\Get	=\\
            \;\;\Let\;(L,W,r) = \Get\,[u,S,(i,C)]\;\In\\
            \;\;\Iif\; (\single\,i' \sqcap ((\unlift\,r)\,x) {=} \bot) \;\Then\;\\
            \;\;\;{\bf \bot}\;\Eelse \Iif\;\accounted\,((i,C),(L,W,r))\,(i',C')\\
            \;\;\Then\; {\bf \bot}\;\Eelse\;\\
            \;\;\Let\;(L',v) = \Get[(i',C')]\;\In\\
            \;\;\Let\;r' = \aSemR{x' \leftarrow^\sharp (\unlift\,v)\,\returnVar}r\;\In\\
            \;\;(\emptyset,(L \sqcup L',W,r'))\\[1ex]
%
            \sem{[u,S,(i,C)],\lock(a),(i',C')}^\sharp\Get =\\
            \;\;\Let\;(L,W,r) = \Get\,[u,S,(i,C)]\;\In\\
            \;\;\Let\;r' = \Iif\;\accounted\,((i,C),(L,W,r))\,(i',C')\;\\
            \quad\Then\;\bot\;\Eelse \bigsqcap_{\Cluster\in\Clusters_a}\Get\,[a,\Cluster,(i',C')]\;\In\\
            \;\;\left(\emptyset,\left(L,W,r \sqcap \left(\left(\bigsqcap_{\Cluster\in\Clusters_a}L\,(a,\Cluster)\right) \sqcup r'\right)\right)\right)\\[1ex]

            \sem{[u,S,(i,C)],g=x}^\sharp\Get	=\\
            \;\;\Let\;(L,W,r) = \Get\,[u,S,A]\;\In\\
			\;\; (\emptyset, (L,W \cup \{g\},\aSemR{g \leftarrow x}\,r))\\[1ex]

            \sem{[u,S,(i,C)],x=g}^\sharp\Get	=\\
            \;\; \Let\;(L,W,r)= \Get\,[u,S,A]\;\In	\\
            \;\; (\emptyset, (L,W,\aSemR{x \leftarrow g}\,r))
        \end{array}
    \]
    \end{minipage}%
    \begin{minipage}[t]{.52\linewidth}\[
        \begin{array}{lll}
            \sem{[u,S,(i,C)],x {=}\create(u_1)}^\sharp\Get	=\\
            \;\;\Let\;(L,W,r) = \Get\,[u,S,(i,C)]\;\In	\\
            \;\;\Let\;(i',C') = \newA\,u\,u_1\,(i,C)\;\In \\
            \;\;\Let\;r' = \restr{(\aSemR{\self \leftarrow^\sharp (\single\,i')}r)}{\X}\;\In\\
            \;\;\Let\;\rho {=} \{[u_1,(\emptyset,(i',C'))] {\mapsto} (L,\emptyset,r') \}\;\In\\
            \;\;(\rho,(L,W,\aSemR{x \leftarrow^\sharp \single\,i'}r))\\[1ex]

        \sem{[u,S,(i,C)],\return\,x,(i,C)}^\sharp\Get	=\\
            \;\;\Let\;(L,W,r) = \Get\,[u,S,(i,C)]\;\In\\
            \;\; \Let\;v = \restr{\left(\aSemR{\returnVar \leftarrow x}\,r\right)}{\{\returnVar\}}\;\In\\
            \;\; \Let\;\rho = \{ [(i,C)] \mapsto (L,v)\}\;\In\\
            \;\; (\rho,(L,W,r))\\[1ex]

        \sem{[u,S,(i,C)],\unlock(a),(i,C)}^\sharp\Get	=\\
            \;\;\Let\;(L,W,r) =\Get\,[u,S,(i,C)]\;\In\\
            \;\;\Let\;(L',\rho) = \Iif\;\GM[a]{\cap} {W} {=} \emptyset\;\Then\; (L,\emptyset)\\
            \;\;\quad\Eelse\;(L \oplus \{(a,\Cluster) \mapsto \restr{r}{\Cluster} \mid \Cluster\in\Clusters_a\},\\
            \;\;\quad\quad \{ [a,\Cluster,(i,C))] \mapsto \restr{r}{\Cluster} \mid \Cluster\in\Clusters_a \})\\
            \;\;\In\\
            \;\;\Let\;r' = \restr{r}{\X \cup \bigcup \{\GM[a'] \mid a' \in (S\setminus{a})\}}\;\In\\
            \;\;\Let\;W' {=} \{ {W} \mid g {\in} W, \MM[g]\cap S{\setminus}\{a\} {\neq} \emptyset\}\\
            \;\;\In\;(\rho,(L',W',r'))
        \end{array}
        \]
    \end{minipage}\end{mdframed}
    \caption{Right-hand sides for the improved (\textbf{I1}, \textbf{I2}) analysis using thread \emph{id}s.}\label{f:improvedAbs}
\end{figure}
\begin{theorem}
    This analysis is sound w.r.t.\ to the local trace semantics.
\end{theorem}
\begin{proof}
\noindent The proof relies on the following observations:
\begin{itemize}
    \item When $\GM[a]\cap W = \emptyset$, no side-effect is required.
	\ignore{ 
	Any thread locking $a$ following this unlock,
        will either consult unknowns that writes were side-effected to, at the unlock $a$ immediately succeeding the write,
        or will have accounted for those values via $L\,a$.
	}
    \item Exclusions based on $\accounted$ are sound, i.e., it only excludes \emph{join-local} writes.
\end{itemize}
The detailed proof is a simplification of a correctness proof for the enhanced analysis from \cref{s:joins}, which we outline in \cref{s:correctness-joins}.\qed
\end{proof}

\noindent The analysis does not make use of components $C$ at unknowns $[a,\Cluster,(i,C)]$ and $[i,C]$.
We detail in \cref{s:not-read-ancestors} how this information
can be exploited to exclude a further class of writes -- namely, those that are performed
by an ancestor of the ego thread before the ego thread was created.
Alternatively,
an implementation may
abandon control-point splitting according to $C$
at mutexes and thread \emph{id}s, replacing the unknowns $[a,\Cluster,(i,C)], [i,C]$ with $[a,\Cluster,i]$ and $[i]$, respectively.

When turning to improvement \textbf{I3}, we observe that after joining
a thread $t$ with a \emph{unique} thread \emph{id}, $t$ cannot perform further writes.
As all writes of joined threads are \emph{join-local} to the ego thread,
it is not necessary to read from the corresponding global unknowns.
We therefore enhance the analysis to also track in the local state, the set $J$ of thread \emph{id}s
for which join has definitely been called in the \emph{join-local} part of the local trace and refine
$\accounted$ to take $J$ into account:
\[
\begin{array}{lll}
\accounted\,((i,C),(J,L,W,r))\,(i',C')=\unique\,i'\land(i=i' \lor i'\in J)
\end{array}
\]
Details on this enhancement can be found in \cref{s:joins}.

\section{Exploiting Clustered Relational Domains}\label{s:clustering}
Naively, one might assume that tracking relations among a larger set of globals is necessarily more precise than 
between smaller sets.
Interestingly, this is no longer true for our analyses, e.g., in presence of thread \emph{id}s. A similar effect where relating more
globals can deteriorate precision has also been observed in the context of an analysis using a data-flow
graph to model interferences~\cite{Farzan12}.
%

\begin{example}
Consider again \cref{e:cluster} in the introduction with $\Clusters_a = \{\{g,h,i\}\}$.
%
%
For this program, the constraint system of the analysis has a unique least solution.
It verifies that
assertion $(1)$ holds.
It assures for $[a,\{g,h,i\},t_1]$ that $h=i$ holds, while for the \emph{main} thread and the program point before each assertion,
$L\,(a,\{g,h,i\}) = \{ g=h, h=i \}$ holds, while
for $[a,\{g,h,i\},\textsf{main}]$ and $[a,\{g,h,i\},t_2]$ only $\top$ is recorded, as is for any relation
associated with $m_g$, $m_h$, or $m_i$.
Assertion (2), however, will not succeed, as the side-effect from $t_1$ causes the older values
from the first write in the main thread to be propagated to the assertions as well, implying
that while $h=i$ is proven, $g=h$ is not.\qed
\end{example}

\noindent
Intuitively, the reason the analysis loses precision is that, at an unlock of mutex $a$,
the current relationships between \emph{all} clusters protected by $a$ are side-effected.
As soon as a single global is written to, the analysis behaves as if \emph{all} protected globals had been written.
A better result thus may be hoped for, if publishing is limited to those clusters for which at least one global has
been written, such that more precise information may remain at others.
%
%

Accordingly, in the improved analysis, when \textbf{unlocking} a mutex $a$, side-effects are only produced to those clusters $\Cluster\in\Clusters_a$
containing at least one global that was written to since the last operation $\lock(a)$.
The definitions for locking and unlocking are given in~\cref{f:clusteredAbs}.

For \textbf{locking} the mutex $a$, on the other hand, the abstract value to be incorporated into the local state is assembled from the contributions of
different threads to the clusters.
In order to allow that, the separate constraints for each admitted digest from \cref{s:refinement} are combined into one constraint
for the set ${\bf I} = \{ (i',C') \mid (i,C) \in \aSem{\lock(a)}_\A ((i,C),(i',C')) \}$ of \emph{all} admitted digests.
This is necessary as side-effects to unaffected clusters at $\unlock(a)$ have been abandoned and thus the meet with the values for clusters
of one thread at a time is unsound.
For each cluster $\Cluster$, the join-local information $L\,(a,\Cluster)$ is joined with all contributions to $\Cluster$
by threads that are not yet \emph{accounted} for, but admitted for $\Cluster$ by the digests.
Here, the contributions of threads that do \emph{not} write $\Cluster$ is $\bot$, and thus do not affect the value for
$\Cluster$.
Finally, the resulting value is used to improve the local state by the operation \emph{meet}.
The right-hand side for $\lock(a)$ thus exploits the fine-grained, per-cluster MHP information provided by the digests and the predicate
$\accounted$.
%
%
%
By construction, we have:
\begin{theorem}\label{t:clusters}
Given domains $\R$ and $\ValD$ fulfilling the requirements from \cref{def:sound},
any solution of the constraint system is sound w.r.t.\ the local trace semantics.
Maximum precision is obtained with $\Clusters_a = 2^{\GM[a]}$.
\qed
\end{theorem}
  For \cref{e:cluster}, with $\Clusters_a = 2^{\GM[a]}$, both
  assertions are verified.
Performing the analysis with all possible subclusters simultaneously can be rather expensive
when sets $\GM[a]$ are large.
The choice of a most appropriate subclustering thus generally involves a trade-off between precision and runtime.
The situation is different if the relational domain is $k$-decomposable:
\begin{theorem}\label{t:optimal}
Provided the relational domain is $k$-decom\-po\-sable (Equations \eqref{def:decomp}), the clustered analysis
using all subclusters of sizes at most $k$ only,
is equally precise as the clustered analysis
using all subclusters $\Clusters_a = 2^{\GM[a]}$ at mutexes $a$.
\end{theorem}
\begin{proof}
Consider a solution $\Get$ of the constraint system with $\Clusters_a = 2^{\GM[a]}$.
Then for unknowns $[a,Q,(i,C)]$ and $[a,Q',(i,C)]$ with $Q\subseteq Q'$ and $\abs{Q} \leq k$, and
values $r{=}\Get\,[a,Q,(i,C)]$, $r'{=}\Get\,[a,Q',(i,C)]$, we have that $r \sqsubseteq \restr{r'}{Q}$
(whenever the smaller cluster receives a side-effect, so does the larger one).
Thus, by $k$-decomposability, the additional larger clusters $Q'$, do not improve the
\emph{meet} over the clusters of size at most $k$ for individual thread \emph{id}s as well as
the \emph{meet} of their \emph{join}s over all thread \emph{id}s.
The same also applies to the clustered information stored in $L$.
%
\qed
\end{proof}
\begin{example}
  Consider again \cref{e:cluster}. If the analysis is performed with clusters
  $\Clusters_a = \{\{h,i\},\{g,h\},\{g,i\},\{g\},\{i\},\{h\}\}$ both assertions can be proven. \qed
\end{example}

\begin{figure}[t]
  \begin{mdframed}
  \begin{minipage}[t]{.57\linewidth}\[
      \begin{array}{lll}
        \sem{[u,S,(i,C)],\unlock(a),(i,C)}^\sharp\Get	=\\
        \;\;\Let\;(L,W,r) =\Get\,[u,S,(i,C)]\;\In\\
        \;\;\Let\;\Clusters' = \{ \Cluster \mid \Cluster\in \Clusters_a, \Cluster \cap W \neq \emptyset \} \;\In\\
        \;\;\Let\;L' = L \oplus \{(a,\Cluster) \mapsto \restr{r}{\Cluster} \mid \Cluster\in \Clusters'\}\;\In\\
        \;\;\Let\;\rho = \{[a,\Cluster,(i,C)] \mapsto \restr{r}{\Cluster} \mid \Cluster\in\Clusters'\}\;\In\\
        \;\;\Let\;r' = \restr{r}{\X \cup \bigcup \{\GM[a'] \mid a' \in (S\setminus{a})\}}\;\In\\
        \;\;\Let\;W' = \{ W \mid g \in W, \MM[g] \cap S\setminus\{a\} \neq \emptyset\}\;\In\\
        \;\;(\rho,(L',W',r''))
      \end{array}
  \]
  \end{minipage}%
  \begin{minipage}[t]{.48\linewidth}\[
      \begin{array}{lll}
        \sem{[u,S,(i,C)],\lock(a),{\bf I}}^\sharp\Get =\\
        \;\; \Let\;(L,W,r) = \Get\,[u,S,(i,C)]\;\In\\
        \;\;\Let\: l = ((i,C),(L,W,r)) \;\In\\
        \;\;\Let\; J(\Cluster) = \bigsqcup  \left\{ \Get\,[a,\Cluster,(i',C')] \mid \right. \\
        \;\;\qquad\left. (i',C') \in {\bf I}, \neg \accounted\;l\,(i',C')\right\}\;\In\\
        \;\;\Let\: r' = \bigsqcap_{\Cluster\in \Clusters_a} \left(J(Q) \sqcup L\,(a,\Cluster) \right)\;\\
        \;\;\In\\
        \;\;(\emptyset,(L,W,r \sqcap r'))
      \end{array}
      \]
  \end{minipage}
  \end{mdframed}
  \caption{Right-hand sides for unlocking and locking when limiting side-effecting to potentially written clusters.}\label{f:clusteredAbs}

\end{figure}
\noindent
The one element clusters, on the other hand, cannot be abandoned -- as indicated by the example from \cref{e:one}
in the additional material.

\section{Experimental Evaluation}\label{s:experiments}
We implemented the analyses extending the context-sensitive static analyzer \textsc{Goblint}.
It supports the dynamic creation of mutexes and provides the set of protecting mutexes for each global.
The implementation tracks information about integral variables using either the \emph{Interval} or the \emph{Octagon} domains from \textsc{Apron}~\cite{JeannetM2009}.
A comparison with other tools is difficult, for details see \cref{a:more-experiments}:
\begin{itemize}
    \item \textsc{Duet}~\cite{Farzan12} ---
	The benchmark suite from the paper is only available as binary goto-programs which neither the current version of \textsc{Duet} nor any other tool considered here can consume.
%
     Since \textsc{Duet} does not support function calls and its inlining support is limited, it could only be run on \emph{some} benchmarks.
    \item \textsc{AstréeA}~\cite{Mine2014} --- A public version is available but not licensed for evaluation.
    \item \textsc{Watts}~\cite{Kusano16} --- Since we were unable to run the tool on any program, we compared with the numbers reported by the authors.
    \item \textsc{NR-Goblint}~\cite{traces1} --- \textsc{Goblint} with the non-relational analyses from \cite{traces1}.
\end{itemize}
We considered four different settings for the analyses, namely, \textit{Interval:} the analysis from \cref{s:base} with Intervals only;
\textit{Octagon:} the analysis from \cref{s:base} with Octagons;
\textit{TIDs:} the analysis from \cref{s:self-excluded} with the enhancement \cref{s:joins} with Octagons;
\textit{Clusters:} \textit{TIDs} using clusters of size at most $2$ only.
All benchmarks were run in a virtual machine on an \textsc{AMD EPYC 7742} 64-Core processor\footnote{The analyzer is single-threaded, so it only used one (virtual) core per analysis job.} running Ubuntu 20.04.
%
The results of our evaluation are summarized in \cref{tab:summary}.

\ignore{
	In particular, we considered the tools \textsc{Duet}~\cite{Farzan12}, \textsc{Sparrow}~\cite{Oh14}, 
	\textsc{Watts}~\cite{Kusano16} and \textsc{AstréeA}~\cite{Mine2014}.
	To our knowledge, there is no public version of \textsc{AstréeA} available for experimentation.
	The \textsc{Sparrow} tool is, at the time of writing, not capable of performing an analysis of multi-threaded code.
	The \textsc{Watts} repository has not received updates in the last five years.
	We tried running \textsc{Watts}, but executing it according to the description provided in the repository failed with an exception.
	We tried to compare to \textsc{Duet} on the benchmark set proposed in \cite{Farzan12}.
	Running the current version of the tool on the benchmarks resulted in difficulties.
	\textsc{Duet} relies on inlining to deal with function calls, but it could not handle calls using function pointers,
	which is necessary to analyze the drivers instrumented with \textsc{DDVerify}.
	However, we executed the tool on our own set of benchmarks, and compared it to different configurations of our analysis.
	The results can be seen in \cref{tab:toy}.
}
\begin{table}[t]
    \centering
    \caption{Summary of evaluation results, with individual programs grouped together.
    For each group the number of programs and the total number of assertions are given.
    \ding{51} indicates all assertions proven, \ding{55} none proven, otherwise the number of proven assertions is given.
    (---) indicates invalid results produced.}
    \label{tab:summary}
    \resizebox{\textwidth}{!}{\newcommand{\s}{\cellcolor{green!30}\ding{51}}%
\newcommand{\e}{\cellcolor{red!30}\ding{55}}%
\newcommand{\w}[1]{\cellcolor{yellow!30}{#1}}%
\newcommand{\invalid}{---}
\begin{tabular}{llcccccc@{\hspace{1.5em}}cc}%
    \toprule
    & & & & \multicolumn{4}{@{}c@{\hspace{1.5em}}}{Our analyzer} & \\
    \cmidrule(lr{1.6em}){5-8}
    Set & Group & \# & Asserts & \parbox[c]{1.1cm}{\centering Interval \\ (Sec.~\labelcref{S:BASE})} & \parbox[c]{1.2cm}{\centering Octagon \\ (Sec.~\labelcref{S:BASE})} & \parbox[c]{1cm}{\centering TIDs \\ (Sec.~\labelcref{S:SELF-EXCLUDED})} & \parbox[c]{1.2cm}{\centering Clusters \\ (Sec.~\labelcref{s:clustering})} & \parbox[c]{2cm}{\centering \textsc{NR-Goblint} \\ w/ interval} & \textsc{Duet} \\
    \midrule
    Our & Basic & 3 & 4 & \s & \s & \s & \s & \s & \w{3} \\
    & Relational & 10 & 35 & \e & \s & \s & \s & \w{4} & \w{2} \\
    & TID & 12 & 19 & \e & \e & \s & \s & \e & \w{2} \\
    & Cluster & 2 & 3 & \e & \e & \w{1} & \s & \e & \w{1} \\
    \addlinespace
    \textsc{Goblint} & POSIX & 5 & 1679 & \w{1146} & \w{1490} & \s & \s & \w{1582} & \invalid \\
    & SV-COMP & 7 & 360 & \s & \s & \s & \s & \s & \invalid \\
    \addlinespace
    \textsc{Watts} & Created & 3 & 3 & \w{2} & \w{2} & \w{2} & \w{2} & \w{2} & \e \\
    & SV-COMP & 5 & 5 & \w{1} & \w{1} & \w{1} & \w{1} & \w{1} & \e \\
    & LKMPG & 1 & 2 & \e & \e & \e & \e & \e & \e \\
    & \textsc{DDVerify} & 28 & 1071 & \w{1043} & \w{1043} & \s & \s & \w{1043} & \invalid \\
    & Scalability & 5 & 740 & \w{735} & \w{735} & \s & \s & \w{735} & \invalid \\
    \addlinespace
    \textsc{Ratcop} & & 19 & 34 & \w{4} & \w{14} & \w{18} & \w{18} & \w{6} & \w{4} \\ 
    \bottomrule
\end{tabular}

}
\end{table}
\paragraph{Our benchmarks.}
To capture particular challenges for multi-threaded relational analysis, we collected a set of small benchmarks (including the examples from this paper) and
added assertions.
On these, we evaluated our analyzer, \textsc{NR-Goblint}, and \textsc{Duet}.
Our analysis in the \textit{Clusters} configuration is capable of verifying all the programs. The other tools could only prove a handful of relational assertions.

\paragraph{\textsc{Goblint} benchmarks~\cite{traces1}.}
These benchmarks do not contain assertions. To still relate the precision of our analyzer to the non-relational analyzer \textsc{NR-Goblint}
and to \textsc{Duet}, we used our tool in the \textit{Clusters} setting to automatically derive invariants at each locking operation.
Perhaps surprisingly, \textsc{NR-Goblint} could verify 95\% of the invariants despite being non-relational and not using thread \emph{id}s.

\paragraph{\textsc{Watts} benchmarks~\cite{Kusano16}.}
These benchmarks from various sources
were instrumented with asserts and significantly changed by the authors.
Our analyses are able to verify all but 7 out of over 1000 assertions.
Due to necessary fixes to benchmark programs, our inability to run their tool, and inconsistent assertion counts in the paper,
the numbers are not directly comparable.
Nevertheless, for their scalability tests, the reported runtimes for \textsc{Watts} are up to two orders of magnitude worse
than ours.
For a more detailed discussion, see \cref{a:more-experiments}.

\paragraph{\textsc{Ratcop} benchmarks~\cite{Mukherjee2017}.}
These were \textsc{Java} programs. After manual translation to C,
our analyzer succeeded in proving all assertions any configuration of \textsc{Ratcop} could with \emph{Octagons}, while
\textsc{Ratcop} required polyhedra in one case.
%
\begin{figure}[t]
	\begin{subfigure}{0.4\textwidth}
		\scalebox{0.87}{
			\begin{tabular}{lrrr}
                \toprule
				Name & LLoC & \parbox[c]{1cm}{\centering \#TIDs (unique)} & \parbox{1.5cm}{
                    ~~~TIDs $\sqsubset$ \hfill \\ \strut\hfill Octagon
                } \\
                \midrule
				pfscan & 550 & 3 (2) & 19.0\%\\
				aget & 581 & 6 (4) & 0.0\% \\
				ctrace & 651 & 3 (3) & 0.0\% \\
				knot & 973 & 9 (5) & 0.0\% \\
				smtprc & 3013 & 2 (2) & 0.8\%\\
				\addlinespace
				iowarrior & 1358 & 4 (4) & 17.1\% \\
				adutux & 1509 & 4 (4) & 0.0\% \\
				w83977af & 1515 & 6 (4) & 12.1\% \\
				tegra20 & 1560 & 7 (5)  & 0.0\% \\
				nsc & 2394 & 11 (7)  & 32.2\%\\
				marvell1 & 2476 & 6 (5)  & 59.5\% \\
				marvell2 & 2476 & 6 (5) & 58.4\%\\
                \bottomrule
			\end{tabular}}
		\subcaption{Number of discovered thread \emph{id}s and
        proportion of program points where analysis with thread \emph{id}s is more precise.}\label{f:evalTIDs}
	\end{subfigure}\begin{subfigure}{0.6\textwidth}
        \centering
		\hspace{2em}\pgfplotslegendfromname{benchmark-time-legend} 
		\\[2ex]
		\begin{tikzpicture}
			\pgfplotstableset{
				create on use/benchmark-loc/.style={
					create col/assign/.code={
						\getthisrow{benchmark}\benchmark
						\getthisrow{loc}\loc
						\edef\entry{\benchmark\space(\loc)}
						\pgfkeyslet{/pgfplots/table/create col/next content}\entry
					}
				}
			}
			\pgfplotstableread{
				benchmark	loc	box	oct	tid	cluster12
				ctrace	651	1.16	2.73	4.55	4.80
				pfscan	550	1.61	2.55	2.79	2.50
				aget	581	1.89	2.21	2.27	2.22
				iowarrior	1358	3.06	3.39	3.48	3.45
				w83977af	1515	5.04	5.55	5.45	5.62
				adutux	1509	3.41	4.70	4.87	4.98
                knot	973	2.58	5.55	8.26	8.30
                smtprc	3013	1.75	3.30	3.22	3.13
			}\timetablesmall
			\pgfplotstableread{
				benchmark	loc	box	oct	tid	cluster12
				nsc	2394	36.85	39.48	47.82	47.32
                marvell1	2476	22.14	24.44	23.01	22.57
                marvell2	2476	25.57	27.49	26.17	26.31
                tegra20	1560	14.06	16.98	17.06	16.69
			}\timetablebig

			\begin{groupplot}[
				group style={
					group size=2 by 1,
					horizontal sep=0.9cm,
				},
				height=5.5cm,
				ybar=0pt,
				enlarge y limits=upper,
				ymin=0,
				xtick=data,
				x tick label style={
					rotate=45,
					anchor=east,
				},
				xtick pos=bottom,
				ymajorgrids=true,
			]
				\nextgroupplot[
					width=0.7\textwidth,
					enlarge x limits=0.08,
					ybar legend,
					bar width=3pt,
					legend to name=benchmark-time-legend,
					legend columns=-1,
					symbolic x coords={pfscan,aget,ctrace,knot,iowarrior,adutux,w83977af,smtprc},
					ytick distance=2,
					ylabel={Analysis time [s]},
					ylabel near ticks,
				]
					\addplot table [x=benchmark,y=box] {\timetablesmall};
					\addlegendentry{Interval};
					\addplot table [x=benchmark,y=oct] {\timetablesmall};
					\addlegendentry{Octagon};
					\addplot table [x=benchmark,y=tid] {\timetablesmall};
					\addlegendentry{TIDs};
					\addplot table [x=benchmark,y=cluster12] {\timetablesmall};
					\addlegendentry{Clusters};
				\nextgroupplot[
					width=0.45\textwidth,
					enlarge x limits=0.20,
					bar width=3pt,
					symbolic x coords={tegra20,nsc,marvell1,marvell2},
					ytick distance=10,
				]
					\addplot table [x=benchmark,y=box] {\timetablebig};
					\addplot table [x=benchmark,y=oct] {\timetablebig};
					\addplot table [x=benchmark,y=tid] {\timetablebig};
					\addplot table [x=benchmark,y=cluster12] {\timetablebig};
			\end{groupplot}
		\end{tikzpicture}
		\subcaption{Analysis times.}
		\label{fig:benchmark-time}
	\end{subfigure}
    \caption{Precision and performance evaluation on the \textsc{Goblint} benchmark set.}
\end{figure}
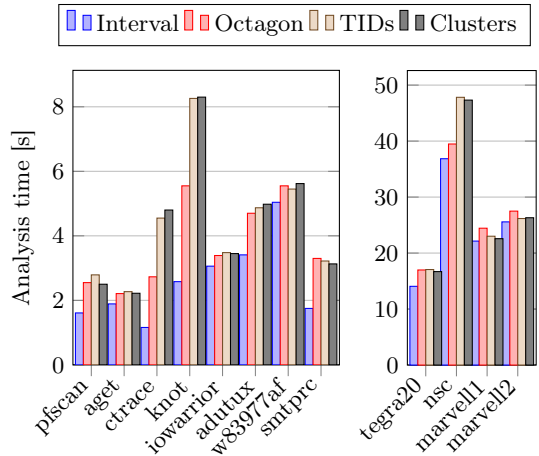
%
%
%
\paragraph{Internal comparison} We evaluated our analyses in more detail on the \textsc{Goblint} benchmark set~\cite{traces1}.
\cref{f:evalTIDs} shows sizes of the benchmark programs (in Logical LoC)
and the number of thread \emph{id}s as found by the analysis from \cref{s:unique}. The high number of threads identified as \emph{unique} is encouraging.
%
%
%
%
For a detailed precision evaluation, we compared the abstract values at each program point (joined over all contexts).
\cref{f:evalTIDs} shows for what proportion of program points precision is increased by tracking thread \emph{id}s.
There were no program points where precision decreased or abstract values became incomparable, while for some programs gains of more than
50\% were observed.
%

\ignore{
	When it comes to the impact of clusters, we performed the analysis with all subclusters of size at most $k=2$,
	as well as without any additional clustering.
	We could not observe any difference in precision, but also no penalty in runtime (see \cref{fig:benchmark-time}).
	As the analysis using clusters of size at most $2$ is no more expensive, but more precise is some cases
	(as indicated in \cref{s:clustering}),
	this suggests using clustering with cluster sizes of (for octagons) at most $2$ in practice.
}
\cref{fig:benchmark-time} illustrates the runtimes of these analyses.
In 9 out of 12 cases, performance differences between our relational analyses are negligible.
In all the cases, using clusters incurs no additional cost.
Therefore, the more precise analysis using clusters of size at most $2$ seems to be the method of choice
for thread-modular relational abstract interpretation.

\section{Related Work}\label{s:related}
%
Since its introduction by Min\'e \citep{Mine01,Mine06Oct}, the weakly relational numerical domain of \emph{Octagons}
has found wide-spread application for the analysis and verification of programs \cite{Blanchet2003,Cousot09}.
%
%
%
Since tracking relations between \emph{all} variables may be expensive,
pre-analyses have been suggested to identify \emph{clusters} of numerical variables
whose relationships may be of interest \cite{Blanchet2003,Cousot09,Oh2015,Heo2016}.
%
%
%
A \emph{dynamic} approach to decompose relational domains into non-overlapping clusters based
on learning is proposed by \citet{Singh2018}.
%
While these approaches trade (unnecessary) precision for efficiency,
others try to \emph{partition} the variables into clusters without compromising precision
\cite{Singh2017,Cousot2019, Halbwachs2003, Halbwachs2006,Oh2014, Singh2018POPL}.
These types of clustering are orthogonal to our approach and could, perhaps, be combined with it.

The integration of relational domains into thread-modular abstract interpretation was pioneered by \citet{Mine2014}.
His analysis is based on \emph{lock invariants} which determine for each mutex a relation which
is meant to hold whenever the mutex is \emph{not} held.
\emph{Weak interferences} then are used to account for asynchronous accesses to variables.
For practical analyses, a relational abstraction only for the lock invariants is proposed,
while using a coarse, non-relational abstraction for the weak interferences.
This framework closely follows the corresponding framework for non-relational analysis \cite{Mine2012},
while abandoning \emph{background locksets}.
Our relational analysis, on the other hand, maintains at each mutex $a$ only relations between variables write-protected by $a$.
For these relations more precise results can be obtained, since they are incorporated into the local
state at locks by \emph{meet} (while \cite{Mine2014} uses \emph{join}).

\citet{Mine2017} present an analysis framework which we consider orthogonal to our approach.
It is tailored to the verification of algorithms that do not rely on explicit synchronization via mutexes
such as the \emph{Bakery} algorithm.
\citet{Mine2018} extend the previous work to handle weak memory effects (PSO, TSO) by incorporating memory buffers into the thread-local semantics.
The notion of interferences is also used by \citet{Sharma21} for the analysis of programs under the Release/Acquire Memory Model of C11 by additionally tracking
abstractions 
of \emph{modification sequences} for global variables. They consider fixed finite sets of threads only, and do not deal with thread creation or joining.

Earlier works on thread-modular relational analysis rely on \textsc{Datalog} rules to model interferences in the sense of Min\'e
in combination with abstract interpretation applied to
the Data-Flow Graph \cite{Farzan12} 
or the Control-Flow Graph 
\cite{Kusano16} (later extended to weak memory \cite{Kusano17}), respectively.
\ignore{
Both settings
are more specific than ours: \cite{Farzan12} considers parameterized systems (where multiple instances of a given
thread may exist), whereas \cite{Kusano16} assumes a fixed finite set of threads. In our setting, thread \emph{id}s are analyzed to
deal with both cases appropriately.
\cite{Kusano16} treats interferences flow-sensitively, by identifying pairs of interferences that may not coexist in a
single execution.
}
\citet{Botbol17} give a non-thread-modular analysis of multi-threaded programs with message-passing concurrency by encoding the program semantics as a symbolic transducer.

All these approaches have in common that clusters of variables, if there are any, are predefined and not treated specially by the analysis.
This is different in the
thread-modular relational analysis proposed by \citet{Mukherjee2017}.
It propagates information from unlock to lock operations.
It is relational for the locals of each thread, and within \emph{disjoint} subsets of globals, called \emph{regions}.
These regions must be determined beforehand with the specific requirement to satisfy \emph{region-race freedom}.
In contrast, the only extra a priori information required by our analysis, are the sets of (write) protecting mutexes of globals --
which can be computed during the analysis itself.
The closest concept within our approach to a \emph{region} is the set of globals jointly protected by mutexes.
These sets may overlap -- which explicitly is exploited by the analysis.
Like ours, their proof of correctness refers to a thread-local semantics. Unlike ours, it is based on interleavings
and thus overly detailed.
\ignore{
- which is one intricate abstraction of local traces.
Instead, our analysis concentrates on tracking thread \emph{id}s,
thread creation/joining as well as exclusive access to multiple variables.
}
%
The concrete semantics on which our analyses are based,
is a collecting local trace semantics extending the semantics of \citet{traces1}
by additionally taking thread termination and joins into account.
The analyses in \cite{traces1}, however, are non-relational. No refinement via further finite abstractions of local traces, such as
thread \emph{id}s is provided.

The thread \emph{id} analysis perhaps most closely related to ours, is by \citet{Feret05} who computes
\emph{id}s for agents in the $\pi$-calculus as abstractions of the sequences of encountered create edges.
Another line of analysis of concurrent programs deals with determining which critical events may possibly happen in parallel
(MHP)~\cite{Naumovich1999,Barik2006,Agarwal07,Di2015,Zhou2018,Albert12,Albert15,Albert17} to detect possible programming errors like, e.g., data races, or
identifying opportunities for optimization.
%
%
Mostly, MHP analyses are obtained as abstractions of a \emph{global} trace semantics of concurrent programs~\cite{Dwyer1994}.
Here, we apply related techniques for improving thread-modular analyses --
but based on a \emph{local} trace semantics.
Like MHP analyses, we take thread creation and joining histories into account
as well as sets of currently held mutexes.
Additionally, we also consider crucial aspects of the modification history of globals
and provide a general framework to obtain further refinements.
%

In a sequential setting, splitting control locations according to some abstraction of reaching traces
is a common technique for improving the precision of dataflow analyses \cite{Holley80,Bodik97}
or abstract interpretation \cite{Handjieva98,Mauborgne05,Rival07,Montagu21}.
Control point splitting can be
understood as an instance of the reduced cardinal power domain~\cite{Cortesi13,Giacobazzi99,Cousot79}.
%
%
For the analysis of multi-threaded programs, \citet{Mine2014}
applies the techniques of \citet{Mauborgne05}
to \emph{single} threads, i.e., independently of the actions of all other threads.
Our approach, on the other hand, may take arbitrary properties of \emph{local} traces into account, and thus is more general.

\section{Conclusion and Future Work}\label{s:conclusion}

We have presented thread-modular relational analyses of global variables tailored to decomposable domains.
In some cases, more precise results can be obtained by considering smaller clusters.
For
$k$-decomposable domains, however,
we proved that the \emph{optimal} result can already be obtained by considering clusters of size at most $k$.
We have provided a framework to incorporate finite abstractions of local traces into the analysis. 
Here, we have applied this framework to take creation as well as joining of threads into account,
but believe that it paves the way to seamlessly enhance the precision of thread-modular abstract 
interpretation.
%
The evaluation of our analyses on benchmarks proposed in the literature indicates 
that our implementation is competitive both w.r.t.\ precision and efficiency.
In future work, we would like to experiment with further abstractions of
local traces, perhaps tailored to particular programming idioms,
and also explore the potential of non-numerical $2$-decomposable domains.
%
%
%
%

\FloatBarrier
{

  \renewcommand{\doi}[1]{\textsc{doi}: \href{http://dx.doi.org/#1}{\nolinkurl{#1}}}
  \bibliographystyle{splncs04nat}
  \bibliography{lit}
}
\clearpage
\appendix
\section{Details on the Local Trace Semantics}
We first detail the handling of local variables and guards for the local trace semantics that are omitted in the main text.
Then, we give an example formalism for local traces that extends the one proposed in \cite{traces1} for handling thread joins and returns.
\subsection{Local Variables and Guards}\label{ss:globals}
In our setting, all expressions $r$ occurring as guards or as non-variable right-hand sides of assignments
may contain only \emph{local} variables, not globals.
A corresponding right-hand side function for guards and assignments to local variables can then be given by
\[
\sem{u,A}\,\eta = (\emptyset,\sem{e}(\eta\,[u]))
\]
That is, the computation is entirely local, and no side-effects are triggered.

\subsection{Formalism for Local Traces}\label{s:details}
The concrete concurrency semantics imposes \emph{restrictions} onto when binary actions are defined.
In particular, a binary operation $\sem{e}$ may only be defined for a pair of local traces $t_0$ and $t_1$
if certain parts of $t_0$ and $t_1$ represent the same computation.
To make such restrictions explicit, we recall the concrete representation of local traces presented
in \cite{traces1} and enhance it for returning from and joining of threads.

A \emph{raw} (finite) trace of single thread $i\in\TIDs$ is a sequence
$\lambda =\bar u_0\act_1\ldots$ $\bar u_{n-1}\act_n\bar u_n$ for states
$\bar u_j = (u_j,\sigma_j)$ with
$\sigma_j\,\self = i$, and actions $\act_j\in\Actions$ corresponding to
the local state transitions of the thread $i$ starting in configuration
$\bar u_0$ and executing actions $\act_j$.
In that sequence, every action $\lock(m)$ is assumed to succeed,
and when accessing a global $g$, any value may be read. The same applies to the return value of
an action $x {=} \join(x')$.

One can view $\lambda$ as an acyclic graph whose nodes are the 3-tuples
$(j,u_j,\sigma_j), j= 0,\ldots,n,$ and
whose edges are $((j-1,u_{j-1},\sigma_{j-1}),\act_j,(j,u_j,\sigma_j)), j= 1,\ldots,n$.
Let $V(\lambda)$ and $E(\lambda)$ denote the
set of nodes and edges of this graph, respectively.

Let $\Lambda(i)$ denote the set of all individual raw traces for thread $i$, and
$\Lambda$ the union of all these sets.
\newcommand{\Vs}{{\mathbf V}}

A \emph{raw global trace} is an acyclic graph $\tau=(\Vs,\E)$ where
$\Vs = \bigcup\{V(\lambda_i)\mid i\in I\}$ and $\E = \bigcup\{E(\lambda_i)\mid i\in I\}$
for a set $I$ of thread \emph{id}s and raw local traces $\lambda_i\in\Lambda(i)$.
On the set $\Vs$, we define the \emph{program order}
as the set of all pairs $\bar u\to_{p}\bar u'$ for which there is an edge $(\bar u,\act,\bar u')$ in $\E$.
To formalize our notion of local traces, we extend the program order to a
\emph{causality order} which additionally takes the order in which threads are
created and joined, as well as the order in which mutexes are acquired and released, into account.

For $a\in\M$, let $a^+ \subseteq \Vs$ denote the set of nodes $\bar u$
where an incoming edge labeled $\lock(a)$ exists, i.e., $ \exists x\,(x,\lock(a),\bar u) \in \E$,
and $a^-$ analogously for $\unlock(a)$.
Analogously, let $J$ and $R$ denote the sets of nodes in $\Vs$ having an
incoming edge labeled $x{=}\join(x')$, and $\return\,x$, respectively, for any local variables $x,x'$.
On the other hand, let $C$ denote the set of nodes with an \emph{outgoing} edge labeled $x{=}\create(u_1)$ (for any local variable $x$ and program point $u_1$).
Let $S$ denote the set of minimal nodes w.r.t.\ to $\to_p$, i.e., the points at which threads start and let $\bf 0$ the node $(0,u_0,\sigma_0)$ where $\sigma_0\,\self = 0$.

A \emph{global trace} $t$ then is represented by a tuple $(\tau,\to_c,\to_j,(\to_a)_{a\in\M},(\to_s)_{s\in\SS})$
where $\tau$ is a raw global trace and the relations $\to_c$, $\to_j$, $\to_a$ ($a\in\M$) are
the create, join, and locking orders as introduced in \cref{s:traces}.
The \emph{causality order} $\leq$ of $t$ then is defined as the reflexive and transitive closure
of the union
\[
    \to_p\cup\to_c\cup\to_j\cup\bigcup_{a\in\M}\to_a
\]
These orders need to satisfy the following properties:
\begin{itemize}
\item	\textbf{Causality order} $\leq$ is a partial order with unique least element
	$(0,u_0,\sigma_0)$ where $\sigma_0\,\self=0$;
\item \textbf{Create order}: $\to_c \subseteq C \times (S\setminus \{ \textbf{0}\})$: $\forall s \in (S \setminus \{ \textbf{0}\}): \abs{\{ z \mid z \to_c s \}} = 1$, that is,
    every thread except the initial thread is created by exactly one
    $\create(...)$ action and $\forall x: \abs{\{ z \mid x \to_c z \}} \leq 1$, i.e.,
    each $\create(...)$ action creates at most one thread.
    Additionally, for $((j-1,u_{j-1},\sigma_{j-1}), x {=} \create(v),(j,u_{j},\sigma_{j})) \in \E$ and
    $(j-1,u_{j-1},\sigma_{j-1}) \to_{c} (0,v,\sigma'_0)$
    $\sigma'_0 = \sigma_{j-1} \oplus \{ \self \mapsto i' \}$ for some thread \emph{id} $i'$ where
    $\sigma_j\,x = i'$, i.e., the creating and the created thread agree on the thread \emph{id}
    of the created thread;
\item \textbf{Join order}: $\to_j \subseteq R \times J$: $\forall j' \in J: \abs{\{ z \mid z \to_j j' \}} = 1$
    and $\forall r' \in R: \abs{\{ z | r' \to_j z \}} \leq 1$
    , that is, each join action in the traces joins exactly one thread also appearing in the trace and each thread is joined at most once.
        Additionally, for
    $((j-1,u_{j-1},\sigma_{j-1}), x {=} \join(x'),(j,u_{j},\sigma_{j})) \in \E$ and
    $((j'-1,u_{j'-1},\sigma'_{j'-1}),\\ \return\,y,(j',u_{j'},\sigma'_{j'})) \in \E$ and
    $(j',u_{j'},\sigma'_{j'}) \to_{j} (j,u_{j},\sigma_{j})$:
    $\sigma_{j-1}\,x' = \sigma'_{j'-1}\,\self$, and  $\sigma_j\,x = \sigma'_{j'-1}\,y$, i.e., the thread that is being joined has the thread \emph{id} stored in $x'$ and after the join, the return value is assigned to $x$;
\item \textbf{Locking order}:
    $\forall a \in \M: \to_a \subseteq (a^- \cup \textbf{0}) \times a^+$:
    $\forall x\in (a^- \cup \textbf{0}): \abs{\{ z \mid x \to_a z \}} \leq 1$ and
    $\forall y\in a^+: \abs{\{ z \mid z \to_a y \}} = 1$, that is, for a mutex $a$ every lock is preceded
    by exactly one unlock (or it is the first lock) of $a$, and each unlock is directly followed
    by at most one lock
\end{itemize}
Additionally, we require a consistency condition on values read from globals.
For $((j-1,u_{j-1},\sigma_{j-1}),x = g,(j,u_{j},\sigma_{j})) \in \E$,
there is a \emph{unique} maximal node $(j',u_{j'},\sigma_{j'})$ with respect to $\leq$
such that $((j'-1,u_{j'-1},\sigma_{j'-1}),g = y,(j',u_{j'},\sigma_{j'})) \in \E$
and $(j',u_{j'},\sigma_{j'}) \leq (j-1,u_{j-1},\sigma_{j-1})$.
The uniqueness of such a maximal node is ensured by $\to_{m_g}$, as every write to a global $g$ is immediately
succeeded by an $\unlock(m_g)$ operation of and every read of global $g$ is immediately proceeded by a
$\lock(m_g)$ operation.
Then $\sigma_j\,x = \sigma_{j'-1}\,y$, i.e., the value read for a global is the last value written to it.
%

\noindent
A global trace $t$ is \emph{local} if it has a unique maximal element $\bar u = (j,u,\sigma)$
(w.r.t.\ $\leq$).
For a local trace $t$, we define a function $\sink\,t$ to extract the tuple $(u,\sigma)$, and a function $\loc\,t$ to obtain the
program point $u$ of the sink node.
The function $\last$, on the other hand, extracts the last action $\act$ of the ego thread (if there is any) and returns $\bot$
otherwise. The partial functions $\new\,u\,u_1$ for program points $u$ and $u_1$ and
$\sem{e}$ for control-flow edges $e$ then are defined by extending
a given local trace appropriately.

\section{Soundness Proof for the Analysis from Section \ref{s:base}}\label{s:baseSound}
Let the constraint system for the analysis from \cref{s:base} be called $\C^\sharp$, and let the constraint system \cref{c:concrete}
for the concrete semantics from \cref{s:traces} be called $\C$.
We first modify $\C$ to obtain a new constraint system $\C'$ that has the same unknowns as $\C^\sharp$, except for thread \emph{id}s.
This means that each unknown $[u]$ for program point $u$ is replaced with the set of unknowns $[u,S]$, $S\subseteq\M$, and each
unknown $[a]$ for mutex $a$ is replaced with the set of unknowns $[u,\Cluster]$, $\Cluster\in\Clusters_a$. We remark that there is at
least one such unknown for each mutex $a$, as we require $\Clusters_a \neq \emptyset$.
The unknowns for thread \emph{id}s and mutexes remain unmodified.
Accordingly, the constraint system $\C'$ consists of these constraints:
\begin{equation}
	\begin{array}{llll}
        (\eta,\eta\,[u_0,\emptyset])	&\sqsupseteq	&(\{ [a,\Cluster] \mapsto \init \mid a\in\M,\Cluster\in\Clusters_a \}, \init) \span \\
        (\eta,\eta\,[u',S{\cup} \{a\}])	&\sqsupseteq	& \sem{[u,S],\lock(a)}\,\eta \qquad \qquad & (\text{for } (u,\lock(a),u')\in\E, S\subseteq\M) \\
        (\eta,\eta\,[u',S{\setminus} \{a\}])	&\sqsupseteq	&\sem{[u,S],\unlock(a)}\,\eta & (\text{for } (u,\unlock(a),u')\in\E, S\subseteq\M)\\
        (\eta,\eta\,[u',S])	&\sqsupseteq	&\sem{[u,S],\act}\,\eta & \hspace{-3em} (\text{for } (u,\act,u')\in\E, S\subseteq\M, \text{other actions } $\act$)
	\end{array}
	\label{c:concrete'}
\end{equation}
where the right-hand sides are modified to consult the appropriate unknown $[u,S]$ instead of $[u]$.
Additionally, the right-hand sides for locking and unlocking mutexes are modified to consult and
side-effect to the appropriate $[a,\Cluster]$, $\Cluster\in\Clusters_a$ as follows:
\[
    \begin{array}{lll}
        \sem{[u,S],\lock(a)}\,\Get &=&
          \left(\emptyset,\sem{e}\left(\Get\,[u],\bigcap_{\Cluster\in\Clusters_a} (\Get\,[a,\Cluster])\right)\right)\\
        \sem{u,\unlock(a)}\,\Get
            &=&	\Let\;T = \sem{e}(\Get\,[u])\;\In\;\\
        & &(\{[a,\Cluster]\mapsto T \mid \Cluster\in\Clusters_a\}, T)
    \end{array}
\]

\noindent First, we relate the unique least solutions of $\C$ and $\C'$ to each other.
\newcommand{\GetP}{\Get'}
Let $\Get$ be the unique least solution of $\C$. We construct from it a mapping $\Get'$ from the unknowns of $\C'$
to $2^\T$.
\[
    \begin{array}{llll}
        \GetP\,[u,S] &=& \Get\,[u] \cap \T_S \qquad & \quad  (\text{for } u\in\N, S\subseteq\M)\\
        \GetP\,[i] &=& \Get\,[i] & \quad  (\text{for } i\in\I)\\
        \GetP\,[a,\Cluster] &=& \Get\,[a] & \quad  (\text{for } a\in\M, \Cluster\in\Clusters_a)\\
    \end{array}
\] where $\T_S$ denotes the set of all traces where the ego thread holds the lockset $S$ at the sink node.
\begin{proposition}\label{p:GetGetP}
    $\GetP$ is the least solution of $\C'$ if and only if $\Get$ is the least solution of $\C$.
\end{proposition}
\begin{proof}
    The proof of \cref{p:GetGetP} is by fixpoint induction.
\end{proof}
\noindent Before moving on to the soundness proof, we require the following insight into the analysis from \cref{s:base}.
\begin{proposition}\label{p:restrict}
  Let $\Get^\sharp$ be a solution of $\C^\sharp$ computed by the analysis from \cref{s:base}. Then,
  \[
    \begin{array}{llll}
      \Get^\sharp\,[u,S] &=& \restr{\Get^\sharp\,[u,S]}{\X\cup\{ g \mid g\in\G, \MM\,[g] \cap S \neq \emptyset \}} \qquad &(\text{for } u\in\N, S\subseteq\M) \\
      \Get^\sharp\,[i] &=& \restr{\Get^\sharp\,[i]}{\{\returnVar\}} &(\text{for }i\in \V^\sharp_\I) \\
      \Get^\sharp\,[a,\Cluster] &=& \restr{\Get^\sharp\,[a,\Cluster]}{\Cluster} &(\text{for }a\in \M, \Cluster\in\Clusters_a)
    \end{array}
  \]
\end{proposition}
\begin{proof}
  The proof of \cref{p:restrict} is once more by fixpoint induction using idempotence of the operation \emph{restrict}.
\end{proof}
To relate solutions of $\C'$ and $\C^\sharp$ to each other, we extend the function $x\mapsto t(x)$ that extracts from a
local trace $t$ the value of the local variable $x$ at the sink node, to also extract the \emph{last} value of a global appearing in
$t$.
Let $\Sigma$ the set of all mappings $\sigma: \X \to \V$ assigning values to local variables.
Recall that in a local trace $t$, there is for each global $g$ ever written in $t$,
a unique \emph{last} (w.r.t.\ the causality ordering) write operation $((j-1,u_{j-1},\sigma_{j-1}),g=x,(j,u_j,\sigma_j))$.
Let $\lW_g: \T \to ((\mathbb{N}_0\times\N\times\Sigma) \times \Actions \times (\mathbb{N}_0\times\N\times\Sigma)) \cup \{\bot \}$
be a function to extract such a last write to $g$ from a local trace,
where $\bot$ indicates that $g$ has not been written to.
Let $\lV_g: \T \to \V$ be a function to extract the \emph{last} value of a global appearing in a local trace.
If a write to the global appears in the local trace, this is the value written at \emph{last} write.
Otherwise, the last value is the initial value of the global, i.e., $0$.
\[
\begin{array}{lll}
    \lV_g\,t =
    \begin{cases}
    \sigma_{j-1}\,x &  \text{if } \lW_g\,t = ((j-1,u_{j-1},\sigma_{j-1}),g=x, \_) \\
    0    & \text{if } \lW_g\,t = \bot\\
    \end{cases}
\end{array}
\]
\noindent Then, we define, for a local trace $t\in\T$ and $x\in (\G\cup\X)$
\[
  \begin{array}{lll}
    t(x) = \begin{cases}
      \sigma\,x & \text{if } x\in\X, \sink\,t=(\_,\sigma)\\
      \lV_g\,t & \text{if } x\in\G
    \end{cases}
  \end{array}
\]

\noindent Let us now relate solutions of $\C'$ and $\C^\sharp$ to each other.
To this end, we define for a local trace $t$ the function
$\beta: \T \to (\Vars \to \V)$ by
\[
  \begin{array}{lll}
    \beta\,t = \left\{ x \mapsto t(x) \mid x\in(\X \cup \G) \right\}
  \end{array}
\]
Based on this, we provide dedicated concretization functions for the unknowns of $\C^\sharp$.
For a local state $r\in\R$, we define
\[
  \begin{array}{lll}
  \gamma_{u,S}(r) &=& \{ t \mid t \in \T_S,\loc\,t = u,
    \exists\,v: (\beta\,t\cup \{\returnVar \mapsto v\}) \in \gamma_\R\,r \}\\
  \gamma_{i}(r) &=& \{ t \mid t \in \T, \last\,t = (\return\,x),
    \sigma = \beta\,t, (\sigma\cup\{\returnVar \mapsto \sigma\,x\}) \in \gamma_\R\,r \}\\
  \gamma_{a}(r) &=& \{ t \mid t \in \T, (\last\,t = \unlock(a) \lor t\in\init),
    \exists\,v: (\beta\,t\cup \{\returnVar \mapsto v\}) \in \gamma_\R\,r \}\\
  \end{array}
\]
and remark that these concretization functions are monotonic.
Additionally, we require for the operation $\nu^\sharp: \N \to \N \to \R \to \V_\I^\sharp$
that it is sound w.r.t.\ to the thread \emph{id}s of created threads in the concrete.
\[
  \begin{array}{lll}
    \gamma_{\V_\I^\sharp}(\nu^\sharp\,u\,u_1\,r) \supseteq \{ t'(\self) \mid t' = \new\,u\,u_1\,t, t\in \gamma_{u,S}(r)\}
  \end{array}
\]
\noindent For a solution $\Get^\sharp$ of $\C^\sharp$, we then construct a mapping $\GetP$ by:
\[
  \begin{array}{llll}
    \GetP [u,S] &=& \gamma_{u,S}(\Get^\sharp\,[u,S]) \qquad &(\text{for } u\in\N, S\subseteq\M)\\
    \GetP [i] &=& \bigcup\{\gamma_{i}(\Get^\sharp\,[i^\sharp]) \mid i \in \gamma_\ValD\,i^\sharp\} \qquad &(\text{for } i\in\TIDs)\\
    \GetP [a,\Cluster] &=& \gamma_{a}(\Get^\sharp\,[a,\Cluster]) \qquad &(\text{for } a\in\M, \Cluster\in\Clusters_a)
  \end{array}
\]

\noindent Altogether, the soundness of the constraint system $\C^\sharp$ follows from the
following theorem.

\begin{theorem}\label{t:baseSoundAppendix}
Every solution of $\C^\sharp$ is sound w.r.t.\ the local trace semantics.
\end{theorem}
\begin{proof}
    Recall from \cref{p:GetGetP}, that the least solution of $\C'$ is sound
    w.r.t.\ the local trace semantics as specified by the constraint system $\C$.
    It thus suffices to prove that the mapping $\GetP$ constructed from a solution $\Get^\sharp$ of $\C^\sharp$
    as above, is a solution of the constraint system $\C'$.
    We verify by fixpoint induction that for the $j$-th approximation $\Get^j$
    to the least solution of $\C'$, $\Get^j \subseteq \Get'$ holds.
    To this end, we verify for the start point $u_0$ and the empty lockset, that
    \[
        (\{ [a,\Cluster] \mapsto \init \mid a \in\M, \Cluster\in\Clusters_a \}, \init) \subseteq (\GetP,\GetP\,[u_0,\emptyset])
    \] holds and for each edge $(u,\act,v)$ of the control-flow graph
    and each possible lockset $S$, that
    \[
      \sem{[u,S],\act}'\, \Get^{j-1} \subseteq(\GetP,\GetP\,[v,S'])
    \] holds.

    We exemplify this for the constraints corresponding to program start, locking of a mutex, and returning from a thread.
    The proof proceeds analogously for all other constraints.

    \noindent First, for the start point $u_0$ and the empty lockset:
    \[
        (\{ [a,\Cluster] \mapsto \init \mid a \in\M,\Cluster\in\Clusters_a \}, \init) \subseteq (\GetP,\GetP\,[u_0,\emptyset])
    \]
    First, for the value of the unknown $[u_0,\emptyset]$: $\init\subseteq\GetP\,[u_0,\emptyset]$:
    \[
        \begin{array}{lll}
            \init^\sharp\,\Get &=&
            \Let\;\rho = \left\{ [a,\Cluster] \mapsto \aSemR{\{ g \leftarrow 0 \mid g\in\Cluster\}}\,\top \mid a \in \M, \Cluster\in\Clusters_a \right\}\;\In\\
            && (\rho,\aSemR{\self \leftarrow^\sharp i_{0}}\top)
        \end{array}
    \]
    Let $\Get^\sharp\,[u_0,\emptyset] = r^\sharp$ the value provided by $\C^\sharp$ for the start point and the empty lockset.
    Since $\Get^\sharp$ is a solution of $\C^\sharp$, $\left(\aSemR{\self \leftarrow i_{0}}\top\right)  \sqsubseteq  r^\sharp$.
    By definition,
    \[
    \begin{array}{lll}
    \Get'[u_0,\emptyset] = \gamma_{u_0,\emptyset}(r^\sharp) &=& \{ t \mid t \in \T_\emptyset, \loc\,t = u_0,
    \exists v: (\beta\,t \cup \{ ret \mapsto v\}) \in \gamma_\R\,r^\sharp \}
    \end{array}
    \]
    For every trace $t\in\init$, let $r = \beta\,t \cup \{ \returnVar \mapsto 0 \}$. Then,
    \[ r \in \gamma_\R \left(\aSemR{\self \leftarrow^\sharp i_{0}}\top\right) \subseteq \gamma_\R (r^\sharp) \]
    and thus $t\in\gamma_{u_0,\emptyset}(r^\sharp)=\GetP\,[u_0,\emptyset]$ for all $t\in\init$.

    \noindent Then, for the side-effects we need to verify that for all $a\in\M$, $\Cluster\in\Clusters_a$:
    $\init \subseteq \gamma_a\left(\aSemR{\{ g \leftarrow 0 \mid g\in\Cluster\}}\,\top\right)$, i.e., the side-effect
    in the abstract covers all the side-effects in the concrete.
    By definition,
    \[
        \begin{array}{lll}
        \gamma_{a}\left(\aSemR{\{ g \leftarrow 0 \mid g\in\Cluster\}}\,\top\right) = \{  t \mid t \in \T, \\
            \quad (\last\,t = \unlock(a) \lor t\in\init),\\
            \quad \exists\,v: (\beta\,t\cup \{\returnVar \mapsto v\}) \in
            \gamma_\R\,\left(\aSemR{\{ g \leftarrow 0 \mid g\in\Cluster\}}\,\top\right)\\
        \}
        \end{array}
    \]
    Consider $t\in\init$, and let $\sigma = \beta\,t\cup \{ \returnVar \mapsto 0 \}$. Then $\forall g\in\G: \sigma\,g = 0$, as all globals initially
    have the value $0$. Then $\sigma\in \gamma_\R\left(\aSemR{\{ g \leftarrow 0 \mid g\in\Cluster\}}\,\top\right)$ and thus for all $t\in\init$
    $t\in\gamma_{a}\left(\aSemR{\{ g \leftarrow 0 \mid g\in\Cluster\}}\,\top\right)$ , i.e., all side-effects are accounted
    for.

    \vspace{2em}
    \noindent Next, for locking of a mutex $a$:
    \[
        \begin{array}{lll}
            \sem{[u,S],\lock(a)}^\sharp\Get^\sharp	&=& \Let\,r^\sharp = \Get^\sharp\,[u, S]\;\In\\
            &&    \Let\,r^{\sharp''} = r^\sharp\sqcap\left(\bigsqcap_{\Cluster\in\Clusters_a}\Get^\sharp\,[a,\Cluster]\right)\;\In\\
            && (\emptyset,r^{\sharp''})
        \end{array}
    \]
    Let $r^\sharp = \Get^\sharp[u,S]$ and $r^{\sharp'} = \Get^\sharp[v,S \cup \{a\}]$ be the value provided by $\Get^\sharp$ for the start and
    endpoint of the edge, respectively. Since $\Get^\sharp$ is a solution of $\C^\sharp$, $r^{\sharp''} \sqsubseteq r^{\sharp'}$.

    \noindent Then, by definition:
    \[
      \begin{array}{lll}
      \GetP\,[v,S \cup \{a\}] = \{ t \mid t \in \T_{S \cup \{a\}},\loc\,t = v, \exists\,v': (\beta\,t\cup \{\returnVar \mapsto v'\}) \in \gamma_\R\,r \}
      \end{array}
    \]
    For any trace $t_0\in\Get^{i-1}\,[u,S]$, let $\beta\,t_0=\sigma_0$. By induction hypothesis,
    $\exists v': \sigma_0 \cup \{\returnVar \mapsto v'\} \in \gamma_\R(r^\sharp)$.
    For any trace $t_1\in\bigcap_{\Cluster\in\Clusters_a}\left(\Get^{i-1}\,[a,\Cluster]\right)$, let $\beta\,t_1=\sigma_1$.
    By induction hypothesis,
    $\exists v'': \sigma_1 \cup \{\returnVar \mapsto v''\} \in \bigcap_{\Cluster\in\Clusters_a}
    \left(\gamma_\R\left(\Get^\sharp\,[a,\Cluster]\right)\right)$.

    For any $t' \in \sem{e}(\{t_0\},\{t_1\})$, $\loc(t')= v$ and $t' \in \T_{S \cup \{a\}}$.
    Let $\sigma' = \beta\,t' = \left\{ x \mapsto t'(x) \mid x\in(\X \cup \G) \right\}$.
    Then $\forall x\in\X: \sigma'\,x = \sigma_0\,x$. 
    For all globals $g\in\GM\,[a]$ on the other hand, one of the following holds:
    \begin{itemize}
      \item[(a)] $S\cap\MM[g]\neq\emptyset \land t'(g) = t_0(g) = t_1(g)$.
      \item[(b)] $S\cap\MM[g]=\emptyset \land t'(g) = t_1(g)$.
    \end{itemize}
    Since $\C^\sharp$ maintains the invariant that $r$ at a program point does not contain information about unprotected globals,
    and $r$ for a mutex and cluster does not maintain information for locals or those globals that are not protected by it
    or not part of the cluster (\cref{p:restrict}),
    and the $\sqcap$ operation in $\R$ is sound w.r.t.\ to the intersection of concretizations
    we have:
\[
  \begin{array}{lll}
    \sigma' \cup \{ \returnVar \mapsto v' \} \in
      \gamma_\R\left(r^\sharp\right) \cap
      \bigcap_{\Cluster\in\Clusters_a}\left(\gamma_\R\left(\Get^\sharp\,[a,\Cluster]\right)\right)\\
    \qquad\qquad\qquad\qquad\qquad\subseteq \gamma_\R \left(r^{\sharp}\sqcap \left(\bigsqcap_{\Cluster\in\Clusters_a}\Get^\sharp\,[a,\Cluster]\right)\right)
    = \gamma_\R \left(r^{\sharp'}\right)
  \end{array}
\] Hence, $t' \in \gamma_{v,S\cup\{a\}}r^{\sharp'}$ and thus $t' \in \GetP\,[v,S\cup \{a\}]$.
    We conclude that the return value of $\sem{[u,S],\lock(a)}\,\Get^{i-1}$ is subsumed by $\GetP\,[v,S\cup \{a\}]$.
    Since the constraint causes no side-effects, the claim holds.

  \vspace{2em}
  \noindent Next, for an edge corresponding to a return from a thread.
  \[
    \begin{array}{lll}
    \sem{[u,S],\return\,x}^\sharp\Get^\sharp	&=&
      	\Let\;r^\sharp = \Get^\sharp\,[u,S]\;\In\\
      &&	\Let\;r^{\sharp''} = \restr{\left(\aSemR{\returnVar \leftarrow x}\,r^\sharp\right)}{\{\returnVar\}}\;\In\\
      &&	\Let\;\rho = \left\{ [\unlift\,r^\sharp\,\self] \mapsto r^{\sharp''}\right\}\;\In\\
      &&		(\rho,r^\sharp)
    \end{array}
  \]
  Let $r^\sharp = \Get^\sharp[u,S]$ and $r^{\sharp'} = \Get^\sharp[v,S]$ be the value provided by $\Get^\sharp$ for the start and
  endpoint of the edge, respectively. Since $\Get^\sharp$ is a solution of $\C^\sharp$, $r^{\sharp} \sqsubseteq r^{\sharp'}$.
  For any trace $t\in\Get^{i-1}\,[u,S]$, let $\beta\,t_0=\sigma_0$. By induction hypothesis,
  $\exists v': \sigma_0 \cup \{\returnVar \mapsto v'\} \in \gamma_\R(r^\sharp)$.
  Let $t'=\sem{e}\,\{t\}$. Then $\beta\,t'=\sigma_1=\sigma_0$. Thus,
  \[
    \sigma_1 \cup \{\returnVar \mapsto v'\} \in \gamma_\R(r^\sharp) \subseteq \gamma_\R(r^{\sharp'})
  \] and $t' \in \gamma_{v,S}r^{\sharp'}$ and thus $t' \in \GetP\,[v,S]$.

  Next for the side-effects caused for an individual trace $t$:
  \[
    \begin{array}{l}
      \{[t(\self)]\mapsto\sem{e}(t)\}\qquad\qquad\qquad
      \left\{ [\unlift\,r^\sharp\,\self] \mapsto \restr{\left(\aSemR{\returnVar \leftarrow x}\,r^\sharp\right)}{\{\returnVar\}}\right\}
    \end{array}
  \]
  As before, consider $t'=\sem{e}(t)$ and $\beta\,t'=\sigma_1=\beta\,t=\sigma_0$.
  We have $\exists v': \sigma_0 \cup \{\returnVar \mapsto v'\} \in \gamma_\R(r^\sharp)$.
  Now set $\sigma' = \sigma_0 \cup \{\returnVar \mapsto \sigma_0\,x\}$.
  Then
  \[
    \sigma' \in \gamma_\R\left(\aSemR{\returnVar \leftarrow x} r^\sharp\right) \subseteq
    \gamma_\R\left(\restr{\left(\aSemR{\returnVar \leftarrow x}\,r^\sharp\right)}{\{\returnVar\}}\right) = r^{\sharp''}
  \]
  Let $i= t(\self)$.
  Then $t' \in \gamma_{i}\left(r^{\sharp''}\right)$, i.e., the value that is side-effected in the concrete semantics
  is abstracted by the abstract value side-effected in the abstract semantics.
  What remains to be shown is that all side-effects in the concrete
  are accounted for in the abstract.
  By induction hypothesis, and the requirements on $\R$ and $\ValD$ (\cref{def:sound}),
  we have $i \in \gamma_\ValD(\unlift\,r^\sharp\,\self)$.
  By construction of $\GetP$, and with $\Get^\sharp$ being a solution of $\C^\sharp$ then
  \[
    t' \in \gamma_{i}\left(r^{\sharp''}\right) \subseteq \gamma_{i}\left(\Get^\sharp[\unlift\,r^\sharp\,\self]\right) \subseteq \GetP\,[i]
  \]
  Thus, the side-effect is accounted for and the claim holds for this constraint.

\end{proof}


\section{Complete Constraint System for the Refined Analysis}\label{s:refineComplete}

Consider the abstract constraint system from \cref{s:base}.
The complete definition of the refined constraint system instantiated to the actions
considered in this paper and unknowns for program points enriched with locksets
then is shown in \cref{f:abstractRef}
\begin{figure}[h]
    \[
      \begin{array}{lll}
        (\Get, \Get\,[u_0,\emptyset,A']) &\sqsupseteq& \init^\sharp_{A'}\,\Get\\
        \quad \text{for start point } u_0, A'\in\initA \span \span \\[1.2ex]

        (\Get, \Get\,\left[v, S \cup \{a\}, A'\right]) &\sqsupseteq& \aSem{[u,S,A_0],\lock(a),A_1}\,\Get\\
        \quad \text{for } (u,\lock(a),v)\in\E, S\subseteq\M, A_0\in\A, A_1\in\A, A'\in\semA{u,\lock(a)}\,(A_0,A_1) \span \span\\[1.2ex]

        (\Get, \Get\,\left[v, S \setminus \{a\}, A'\right]) &\sqsupseteq& \aSem{[u,S,A_0],\unlock(a),A'}\,\Get\\
        \quad \text{for } (u,\unlock(a),v)\in\E, S\subseteq\M, A_0\in\A, A'\in\semA{u,\unlock(a)}\,(A_0)  \span \span\\[1.2ex]

        (\Get, \Get\,\left[v, S, A'\right]) &\sqsupseteq& \aSem{[u,S,A_0],x'{=}\join(x),A_1}\,\Get\\
        \quad \text{for } (u,x'{=}\join(x),v)\in\E, S\subseteq\M, A_0\in\A, A_1\in\A, A'\in\semA{u,x'{=}\join(x)}\,(A_0,A_1)  \span \span \\[1.2ex]

        (\Get, \Get\,\left[v, S, A'\right]) &\sqsupseteq& \aSem{[u,S,A_0],\return\,x,A'}\,\Get\\
        \quad \text{for } (u,\return\,x,v)\in\E, S\subseteq\M, A_0\in\A, A'\in\semA{u,\return\,x}\,(A_0) \span \span\\[1.2ex]

        (\Get, \Get\,\left[v,S, A'\right]) &\sqsupseteq& \aSem{[u,S,A_0],\act}\,\Get\\
        \quad \text{for } (u,\act,v)\in\E,S\subseteq\M, \text{other action }\act, A_0\in\A, A'\in\semA{u,\act}\,(A_0), \span \span
    \end{array}
    \]
    \caption{Constraint system with refinement.}
    \label{f:abstractRef}
\end{figure}
where the
right-hand sides of the system refined with control-point splitting, are determined by
the original right-hand sides as follows
\[
\begin{array}{lll}
    \aSem{[u,S,A_0],\lock(a),A_1}\,\Get &=&
    \Let\,\Get'\,[x] = \Iif\;[x] = [u,S]\;\Then\;\Get\,[u,S,A_0]\\
    & &\quad\Eelse\,\Get\,[x,A_1]\;\In \\
    & &\Let\, (\emptyset,v) = \aSem{[u,S],\lock(a)}\,\Get'\;\In\\
    & &(\emptyset,v)\\[1ex]

    \aSem{[u,S,A_0],\unlock(a),A'}\,\Get &=&
    \Let\,\Get'\,[x] = \Get\,[x,A_0]\;\In\\
    & &\Let\, (\rho,v) = \aSem{[u,S],\unlock(a)}\,\Get'\;\In\\
    & &\Let\, \rho' = \left\{ [a,\Cluster,A'] \mapsto v' \mid ([a,\Cluster] \mapsto v') \in\rho \right\}\\
    & &(\rho',v)\\[1ex]

    \aSem{[u,S,A_0],x' {=} \join(x''),A_1}\,\Get &=&
    \Let\,\Get'\,[x''] = \Iif\;[x] = [u,S]\;\Then\;\Get\,[u,S,A_0]\\
    & &\quad\Eelse\,\Get\,[x,A_1]\;\In \\
    &&\Let\, (\emptyset,v) = \aSem{[u,S],x'' {=} \join(x')}\,\Get'\;\In\\
    && (\emptyset,v)\\[1ex]

    \aSem{[u,S,A_0],\return\,x',A'}\,\Get &=&
    \Let\,\Get'\,[x] = \Get\,[x,A_0]\;\In\\
    &&\Let\, (\{ [i] \mapsto v' \},v) = \aSem{[u,S],\return\,x}\,\Get'\;\In\\
    && (\{ [i,A'] \mapsto v' \},v)\\[1ex]

    \aSem{[u,S,A_0],x{=}\create(u_1)}\,\Get &=&
    \Let\,\Get'\,[x] = \Get\,[x,A_0]\;\In\\
    &&\Let\, (\{[u_1,\emptyset] \mapsto v'\},v) = \aSem{[u,S],x{=}\create(u_1)}\,\Get'\;\In\\
    &&(\{ [u_1,\emptyset,A'] \mapsto v' \mid A'\in \newA\,u\,u_1\,A_0\},v)    \end{array}
    \]
    \[
    \begin{array}{lll}

    \init^\sharp_{A'}\,\Get &=&
    \Let\,\Get'\,\_ = \bot\;\In\\
    &&\Let\,(\rho,v) = \init^\sharp\,\Get'\;\In\\
    &&\Let\, \rho' = \left\{ [a,\Cluster,A'] \mapsto v' \mid ([a,\Cluster] \mapsto v') \in\rho \right\}\\
    &&(\rho',v)\\[1ex]
    \end{array}
    \] and, for all other actions $\act$, by
    \[
    \begin{array}{lll}
    \aSem{[u,S,A_0],\act}\,\Get &=&
    \Let\,\Get'\,[x] = \Get\,[x,A_0]\;\In\\
    &&\aSem{[u,S],\act}\,\Get'
\end{array}
\]
\newcommand{\CRefSharp}{\C^{\sharp'}}
Let us call this constraint system instantiated with the analysis from \cref{s:base} $\CRefSharp$.

\section{Soundness Proof for the Refinement from Section \ref{s:refinement}}\label{s:refineSound}
Consider a modified concrete constraint system, where we associate with each unknown also information from $\A$:
\[
        \begin{array}{lll}
        (\eta,\eta\,[u_0,\emptyset,A])	&\sqsupseteq	&
        (\{ [a,\Cluster,A] \mapsto \{ t \mid  t\in\init, A = \alpha_\A(t)\} \mid a\in\M, \Cluster\in\Clusters_a  \},\\
        && \qquad \{ t \mid  t\in\init, A = \alpha_\A(t)\}) \\[1.2ex]

        (\eta,\eta\,[u',S\cup \{a\},A'])	&\sqsupseteq	&  \left(\sem{[u,S,A_0],\lock(a),A_1}\,\eta\right)\\
        \qquad \text{for } (u,\lock(a),u')\in\E, S\subseteq\M, A' \in  \semA{u,\lock(a)}(A_0,A_1)  \span \span \\[1.2ex]

        (\eta,\eta\,[u',S\setminus \{a\},A'])	&\sqsupseteq	& \left(\sem{[u,S,A_0],\unlock(a),A'}\,\eta\right)\\
        \qquad \text{for } (u,\unlock(a),u')\in\E, S\subseteq\M, A' \in \semA{u,\unlock(a)}(A_0) \span \span\\[1.2ex]

        (\eta,\eta\,[u',S\cup \{a\},A'])	&\sqsupseteq	&  \left(\sem{[u,S,A_0],x{=}\join(x'),A_1}\,\eta\right)\\
        \qquad \text{for } (u,x{=}\join(x'),u')\in\E, S\subseteq\M, A' \in  \semA{u,x{=}\join(x')}(A_0,A_1)  \span \span \\[1.2ex]

        (\eta,\eta\,[u',S,A'])	&\sqsupseteq	& \left(\sem{[u,S,A_0],\act}\,\eta\right)\\
        \qquad \text{for } (u,\act,u')\in\E, S\subseteq\M,  A' \in \semA{u,\act}(A_0), \text{other actions } $\act$ \span \span
        \end{array}
        \label{c:concrete''}
\]
The right-hand sides are modified to access the corresponding unknowns with appropriate digests
and also re-direct side-effects accordingly in a similar manner to what was done in the abstract in \cref{s:refineComplete}.
We exemplify this for the locking of a mutex $a$ here:
\newcommand{\GetPP}{\Get''}
\[
\begin{array}{lll}
    \sem{[u,S,A_0],\lock(a),A_1}\,\GetPP &=&
     \Let\,\Get'\,[x] = \Iif\;[x] = [u,S]\;\Then\;\GetPP\,[u,S,A_0]\\
    &&\quad\Eelse\;\;\GetPP\,[x,A_1]\;\In \\
    &&\Let\, (\emptyset,v) = \sem{[u,S],\lock(a)}\,\Get'\;\In\\
    && (\emptyset,v)
\end{array}
\]
Let this constraint system be called $\C''$.

Let $\GetP$ be the unique least solution of constraint system $\C'$ from \eqref{c:concrete'}. We construct from
it a mapping $\GetPP$ from the unknowns of $\C''$ to $2^\T$.

\[
    \begin{array}{llll}
        \GetPP\,[u,S,A] &=& \GetP\,[u,S] \cap \T_A  \qquad \qquad & (u\in\N, S\subseteq\M, A\in\A)\\
        \GetPP\,[i,A] &=& \GetP\,[i] \cap \T_A & (i\in\I, A\in\A)\\
        \GetPP\,[a,\Cluster,A] &=& \GetP\,[a,\Cluster] \cap \T_A & (a\in\M, \Cluster\in\Clusters_a A\in\A)\\
    \end{array}
\] where $\T_A$ denotes the subset of local traces where $\{ t \mid t\in\T, \alpha_\A(t)=A \}$.
\begin{proposition}\label{p:GetPGetPP}
    Provided that $\alpha_\A$, $\newA$, and $\semA{u,\act}$
    fulfill the requirements \eqref{def:Asound}, \eqref{def:AnewSound}, and \eqref{def:AinitSound}
    from \cref{s:refinement},
    $\GetPP$ is the least solution of $\C''$ if and only if $\GetP$ is the least solution of $\C'$.
\end{proposition}
\begin{proof}
    The proof of \cref{p:GetPGetPP} is by fixpoint induction.
\end{proof}
\noindent Now let us relate solutions of the refined abstract constraint system $\CRefSharp$ and the refined concrete constraint system $\C''$.
Let $\Get^{\sharp'}$ a solution of $\CRefSharp$. We construct from it a mapping $\GetPP$ by reusing the definition of concretizations $\gamma$ from
\cref{s:baseSound}:
\[
  \begin{array}{llll}
    \GetPP [u,S,A] &=& \gamma_{u,S}(\Get^{\sharp'}\,[u,S,A]) \cap \T_A \qquad \qquad & (u\in\N, S\subseteq\M, A\in\A) \\
    \GetPP [i,A] &=& \bigcup\left\{\gamma_{i}((\Get^{\sharp'}\,[i^\sharp,A])) \mid i \in \gamma_\ValD\,i^\sharp\right\} \cap \T_A \qquad & (i\in\TIDs, A\in\A) \\
    \GetPP [a,\Cluster,A] &=& \gamma_{a}((\Get^{\sharp'}\,[a,\Cluster])) \cap \T_A \qquad & (a\in\M,\Cluster\in\Clusters_a, A\in\A)
  \end{array}
\]

\begin{theorem}\label{t:refine}
    Assume that $\alpha_\A$, $\newA$, and $\semA{u,\act}$ fulfill the requirements \eqref{def:Asound}, \eqref{def:AnewSound}, and \eqref{def:AinitSound}.
    Then any solution of the refined constraint system is sound relative to the collecting local trace semantics.
\end{theorem}
\begin{proof}
    By \cref{p:GetPGetPP} any solution of $\C''$ is sound w.r.t. to $\C'$, and from \cref{p:GetGetP} that any solution of $\C'$ is sound w.r.t.\
    to the collecting local trace semantics as specified by $\C$.
    It thus suffices to prove that the mapping $\GetPP$ as constructed from a solution of $\CRefSharp$ is a solution of $\C''$.

    We verify by fixpoint induction that for the $j$-th approximation $\Get^j$
    to the least solution of $\C''$, $\Get^j \subseteq \GetPP$ holds.
    To this end, we verify for the start point $u_0$ and the empty lockset, and all $A \in \A$
    \[
        \begin{array}{lll}
        (\{ [a,\Cluster,A] \mapsto \{ t \mid  t\in\init, \Cluster\in\Clusters_a, A = \alpha_\A(t)\}  \},\{ t \mid  t\in\init, A = \alpha_\A(t)\}) \\
        \qquad \subseteq (\GetPP,\GetPP\,[u_0,\emptyset,A])
        \end{array}
    \] and for each edge $(u,\act,v)$ of the CFG
    and each possible lockset $S$ and digest $A$, that
    \[
      \sem{[u,S,A_0],\act}\, \Get^{j-1} \subseteq(\GetPP,\GetPP\,[v,S',A'])
    \] holds (where additional arguments $A_1$ and $A'$ are passed to the right-hand sides as required by the constraint system).
    We exemplify this for an edge corresponding to a \emph{lock} operation for some mutex $a$.
    We have
    \[
        \begin{array}{lll}
            \sem{[u,S,A_0],\lock(a),A_1}\,\GetPP &=&
             \Let\,\Get_c\,[x] = \Iif\;[x] = [u,S]\;\Then\;\GetPP\,[u,S,A_0]\\
            && \quad\Eelse\;\GetPP\,[x,A_1]\;\In \\
            &&\Let\, (\emptyset,v) = \sem{[u,S],\lock(a)}\,\Get_c\;\In\\
            && (\emptyset,v)\\[1ex]

            \aSem{[u,S,A_0],\lock(a),A_1}\,\Get^{\sharp} &=&
            \Let\,\Get^\sharp_a\,[x] = \Iif\;[x] = [u,S]\;\Then\;\Get^{\sharp}\,[u,S,A_0]\\
            &&\quad\Eelse\;\Get^\sharp\,[x,A_1]\;\In \\
            &&\Let\, (\emptyset,v^\sharp) = \aSem{[u,S],\lock(a)}\,\Get^\sharp_a\;\In\\
            && (\emptyset,v^\sharp)\\[1ex]
        \end{array}
    \]
    \noindent By induction hypothesis, $\Get^{j-1} \subseteq \GetPP$.
    Then, by construction for $\Get_c$ 
    and $\Get^{\sharp}_a$
    and every unknown $x$ of the non-refined constraint system,
    \[
      \Get_c\,x \subseteq \gamma_x(\Get^{\sharp}_a\,x)
    \]
    and thus, by \cref{t:baseSoundAppendix}
    \[
     v \subseteq \gamma_{v,S\cup\{a\}}(v^\sharp)
    \]
    Since $\Get^\sharp$ is a solution of $C^{\sharp'}$, $v^\sharp \sqsubseteq \Get^\sharp\,[v,S\cup\{a\},A']$.
    Also, by \cref{def:Asound} $\left|\semA{u,\lock(a)}(A_0,A_1)\right| \leq 1$ holds. Altogether, then
    \[
        \begin{array}{lll}
        v = v \cap T_{A'} \subseteq \gamma_{v,S\cup\{a\}}(v^\sharp) \cap T_{A'} \subseteq \gamma_{v,S\cup\{a\}}(\Get^\sharp\,[v,S\cup\{a\},A']) \cap T_{A'} = \GetPP\,[v,S\cup\{a\},A']
        \end{array}
    \]
\noindent Since the constraint causes no side-effects, the claim holds.
The proof proceeds analogously for all other constraints.
\end{proof}

\section{Further Increasing Precision of the Analysis from Section~\ref{s:self-excluded}}\label{s:not-read-ancestors}
The analysis as presented in \cref{s:self-excluded} does not make use of components $C$ forming part of
unknowns $[a,\Cluster,(i,C)], [i,C]$ ($a$ a mutex, $\Cluster\in\Clusters_a$, $i$ an abstract thread \emph{id}).
This information can be exploited to exclude a further class of writes -- namely, those that are performed
by a creating thread before the ego thread or its ancestor was created.
Any writes that the created thread may read from the creating thread before the created thread
is created are already accounted for in the start state of the created thread, so only those that happened after the
creation of the current thread need to be considered. To this end, one sets
\[
\begin{array}{lll}
\accounted\,((i,C),(L,W,r))\,(i',C') = (\unique\,i \land i = i') \lor\\
\quad (\lcommondefinite\,i'\,i = i' \land \not\exists \angl{u,u'} \in C': ((i' \circ \angl{u,u'}) = i) \lor \maycreate\,(i' \circ \angl{u,u'})\,i)
\end{array}
\]
\begin{example}\label{e:opts}
    Consider the following program where $\MM[g] = \{a,m_g\}$ and $\MM[h] = \{a,m_h\}$ and assume
    $\Clusters_a = \{\{g,h\}\}$. 
    \begin{center}
        \begin{minipage}[t]{4.5cm}
        \begin{minted}{c}
        main:
         |\op{lock}|(a);
         g = 5; h = 8;
         |\op{unlock}|(a);
         |\op{lock}|(a);
         g = 10; h = 10;
         |\op{unlock}|(a);
         x = |\op{create}|(t1);
         |\op{lock}|(a);
         g = 20; h = 20;
         |\op{unlock}|(a);
         y = |\op{join}|(x);
         |\op{lock}|(a);
         // ASSERT(g==20); (2)
         |\op{unlock}|(a);
        \end{minted}
        \end{minipage}
        \begin{minipage}[t]{3.8cm}
        \begin{minted}{c}
        t1:
         |\op{lock}|(a);
         // ASSERT(g==h); (1)
         |\op{unlock}|(a);
         return 0;
        \end{minted}
        \end{minipage}
    \end{center}
    Here, if we make use of the additional information on which create edges have been encountered, the analysis can determine that assertion (1)
    in thread $t1$ holds.\qed
\end{example}
\noindent The analysis still may incur an unnecessary loss of precision when past writes of a thread are propagated to
a thread it creates, and then back to the creating thread upon \emph{join}.
This may happen when the created thread has not
written to some globals, but the original thread has overwritten them locally in the meantime.
\begin{example}
Consider again \cref{e:opts}.
Here, the assertion (2) cannot be proven, as the stale information in $L\,(a,\{g,h\})$ in $t_1$ (which includes the fact that $g$ may be $10$) is incorporated
into the main thread upon join.
This can be prevented by tracking, for each thread, the set $\bar W$ of global variables that
\emph{may} have been written by it or any thread it has joined,
and then only joining in the $L$ information of the joined thread if at least one protected global has been written.
\qed
\end{example}
\noindent We do not detail this improvement here, but use it in our implementation.

Further useful abstractions to maintain in the ego thread may, e.g., track for each created thread $t'$, the set of globals that has been
potentially written to by the \emph{join-local} part of the ego thread since the creation $W_C\,t'$.
Then, upon joining $t'$, for mutexes $a$ and clusters $\Cluster \in \Clusters_a$ where $\Cluster \cap W_C\,T' = \emptyset$,
the $L\,(a,\Cluster)$ information of the joined thread definitely contains the most up-to-date information,
and $L\,(a,\Cluster)$ of the ego thread can be discarded.

\section{Analysis and Exploiting of Thread Joins}\label{s:joins}
We detail here the enhanced analysis that also tracks in the local state, the set $J$ of thread \emph{id}s
for which join has definitely been called in the \emph{join-local} part of the local trace.
Recall the refined definition of $\accounted$
that takes $J$ into account:
\[
\begin{array}{lll}
\accounted\,((i,C),(J,L,W,r))\,(i',C')=\unique\,i'\land(i=i' \lor i'\in J)
\end{array}
\]
We proceed to give the \emph{interesting} right-hand sides, all others simply propagate this value locally,
and pass it as an additional argument to $\accounted$.
\[
\begin{array}{lll}
\sem{[u,S,(i,C)],\return\,x,(i,C)}^\sharp\Get	&=&
     \Let\;(J,L,W,r) = \Get\,[u,S,(i,C)]\;\In\\
    && \Let\;v = \restr{\left(\aSemR{\returnVar \leftarrow x}\,r\right)}{\{\returnVar\}}\;\In\\
    && \Let\;\rho = \{ [(i,C)] \mapsto (J,L,v)\}\;\In\\
    && (\rho,(L,W,r))
\end{array}
\]

\[
\begin{array}{lll}
    \sem{[u,S,(i,C)],x' {=} \join(x),(i',C')}^\sharp\Get	&=&
        \Let\;(J, L,W,r) = \Get\,[u,S,(i,C)]\;\In\\
        &&\Iif\;(i' \sqcap ((\unlift\,r)\,x) = \bot) \lor\\
        &&\quad \accounted\,((i,C),(J,L,W,r))\,(i',C')\\
        &&\Then\; {\bf \bot}\\
        &&\Eelse\; \Let\;(J',L',v) = \Get[(i',C')]\;\In\\
        &&\quad\Let\;r' = \aSemR{x' \leftarrow^\sharp (\unlift\,v)\,\returnVar}r\;\In\\
        &&\quad(\emptyset,(J \cup J' \cup \{i'\},L \sqcup L',W,r'))\\

    %
    \end{array}
\]
We remark that, in the analysis, when performing a thread join, if there are different thread \emph{id}s for which join might be called, there is
one constraint for each, and the resulting values will be joined to obtain the new value after the join edge. Here, the lattice \emph{join}
for $J$ is intersection.
This handles the case where multiple threads are joined by different threads for which the thread \emph{id}s may be stored in $x$.

Recall that the function $\accounted$ now is given by:
\[
\begin{array}{lll}
\accounted\,((i,C),(J,L,W,r))\,(i',C') =\quad \unique\,i'\land(i=i' \lor i'\in J)
\end{array}
\]

\begin{example}\label{e:joins}
    Consider the following program where $\MM[g] = \{a,m_g\}$ and $\MM[h] = \{a,m_h\}$ and assume
    $\Clusters_a= \{\{g,h\}\}$.
    \begin{center}
        \begin{minipage}[t]{4.5cm}
        \begin{minted}{c}
        main:
         x = |\op{create}|(t1);
         |\op{lock}|(a);
         g = 20; h = 20;
         |\op{unlock}|(a);
         y = |\op{join}|(x);
         |\op{lock}|(a);
         // ASSERT(g==h); (1)
         g = 5; h = 5;
         |\op{unlock}|(a);
         |\op{lock}|(a);
         // ASSERT(g==5); (1)
         |\op{unlock}|(a);
        \end{minted}
        \end{minipage}
        \begin{minipage}[t]{3.8cm}
        \begin{minted}{c}
        t1:
         |\op{lock}|(a);
         g = 4; h = 8;
         |\op{unlock}|(a);
         x = ?;
         |\op{lock}|(a);
         g = x; h = x;
         |\op{unlock}|(a);
         return 0;
        \end{minted}
        \end{minipage}
    \end{center}
    Here, both assertions can be proven. At (1) the thread $t_1$, is must-joined. Its last write is accounted for in $L\,(a,\{g,h\})$,
    thus the unknown $[a,\{g,h\},t_1]$ where the abstract relationship $g=h$ does not hold is not consulted.
    As the updates in $L$ are destructive, after the \emph{main} thread writes $5$ to $g$, this is also the only value it reads
    for $g$, meaning (2) is proven as well.
    \qed
\end{example}

To gain additional precision, the set $J$ of must-joined threads could be published together with the protected globals at an unlock.
In this way, a thread need not read from another thread, that does all its writes after the first thread has already been
\emph{must} joined.
A further increase in precision may be obtained by tracking $J$ as an additional component in $\A$.


\section{Soundness Proof of Analysis with Joins}\label{s:correctness-joins}
The proof proceeds in the following manner.
We provide a modified instance of $\A$ that tracks, in addition to the information from \cref{s:unique}, for each mutex
the information of which abstract thread id (computed in the same manner as in \cref{s:unique}) did the last unlock
immediately succeeding a thread-local write to a global protected by that mutex, and how many such unlocks have happened since the start of the program.
For that, we choose $\A = (I^\sharp\times2^P) \times (\M \to (\mathbb{N}_0\times(I^\sharp\times2^P))$.
For convenience, we here (intermediately) allow control-point splitting even for an \emph{infinite} set.
The abstract function $\semA{u,\lock(a)}$ for a mutex $a$, e.g., is given by
\[
  \begin{array}{l}
    \semA{u,\lock(a)}\,((i,C),H)\,((i',C'),H') = \\[1ex]
    \;\; \begin{cases}
      \emptyset \qquad\qquad\qquad\qquad\qquad\qquad\qquad\qquad\qquad\qquad \neg \textsf{can\_be\_started}\,(i,C)\,(M'\,a)_2 \span \\
      \left\{\left(\begin{array}{lll}
        (i,C),\\
        \{ a \mapsto H\,a \mid a\in\M, (M\,a)_1 \geq (H'\,a)_1 \}\\
        \cup \{ a \mapsto H'\,a \mid a\in\M, (M\,a)_1 < (H'\,a)_1 \}
      \end{array}\right)\right\}
        & \text{else}

    \end{cases}
  \end{array}
\]
where $(\cdot)_k$ is shorthand for accessing the $k$-th component of a tuple.
The concretization for abstract values at unknowns at a mutex $[a,\Cluster,((i,C),H)]$ is defined such that it contains
any local trace ending in an $\unlock(a)$
in which the last thread-local write to a global in $\GM\,[a]$, was by the thread with the thread \emph{id} $(H\,a)_2$.

For local traces, we again introduce a function
$\beta: \T \to (2^\I \times (\M \to (\Vars\to\V) \times 2^\G \times (\Vars \to \V))$ which we will then use
to define appropriate concretization functions.
Let $\lTlW_g: \T \to ((\mathbb{N}_0\times\N\times\Sigma) \times \A \times (\mathbb{N}_0\times\N\times\Sigma)) \cup \{\bot \}$ be a function
to extract the last thread-local write to $g$ if it exists, and return $\bot$ otherwise.

If there is a thread-local lock for a mutex $a$, there also is a \emph{last} thread-local lock of $a$.
Let $\lTlL_a: \T \to ((\mathbb{N}_0\times\N\times\Sigma) \times \A \times (\mathbb{N}_0\times\N\times\Sigma)) \cup \{\bot \}$ be a function
to extract this last thread-local lock of $a$ if it exists, and return $\bot$ otherwise.
Analogously for $\lTlU_a$.

Let  $\jLj: \T \to 2^{((\mathbb{N}_0\times\N\times\Sigma) \times \A \times (\mathbb{N}_0\times\N\times\Sigma))}$ be a function
to extract from a local trace the set of all calls to join that are \emph{join-local} to it.

Let $\lJluW_a: \T \to ((\mathbb{N}_0\times\N\times\Sigma) \times \A \times (\mathbb{N}_0\times\N\times\Sigma)) \cup \{\bot\}$ be a function
to extract, for a mutex $a$, the first \emph{join-local} $\unlock(a)$ action that immediately succeeds the last
\emph{join-local} write to a global in $\GM\,[a]$. If mutex $a$ is currently held, $\lJluW_a$ only considers the subtrace that ends
with the ego-thread acquiring $a$. 

Let $\luW_a: \T \to ((\mathbb{N}_0\times\N\times\Sigma) \times \A \times (\mathbb{N}_0\times\N\times\Sigma)) \cup \{\bot\}$ be a function
to extract, for a mutex $a$, the first $\unlock(a)$ action that immediately succeeds the last
write to a global in $\GM\,[a]$.

Also, for a node $\bar u'$ appearing in a local trace $t$, let $(\bar u')\downarrow_t$ the local sub-trace ending in that node.

\[
  \begin{array}{lll}
    \beta\,t &=& (J,L,W,r) \text{ where }\\[1ex]
    J &=& \{ \sigma_i\,x' \mid  ((i,u_i,\sigma_i),x{=}\join(x'),\_) \in \jLj\,t \}\\
    L &=& \{ (a,\Cluster) \mapsto \left\{ x \mapsto t'(x) \mid g\in (\G\cup\X) \right\} \mid  a\in\M, \Cluster\in\Clusters_a,\\
    && \qquad \lJluW_a\,t=\bot, t' = (u_0,0,\sigma_0)\downarrow_t  \}\\
    && \cup \{ a \mapsto \left\{ x \mapsto t'(x) \mid g\in (\G\cup\X) \right\} \mid  a\in\M,\\
    && \lJluW_a\,t=(\bar u_i,\_,\_),  t' = (\bar u_i)\downarrow_t \}  \\
    W &=& \bigcup_{a\in\M} \{ g \mid g\in\MM[a], \lTlW_g\,t = (\bar {u_i},\_,\_),\\
    && \qquad ((\lTlU_a\,t = (\bar {u_j},\_,\_) \land \bar {u_j} \leq  \bar {u_i} ) \lor \lTlU_a\,t = \bot)  \}\\
    r &=& \left\{ x \mapsto t(x) \mid x\in(\X \cup \G) \right\}
  \end{array}
\]
Based on this, we provide dedicated concretization functions for the unknowns of the constraint system:
\[
  \begin{array}{lll}
  \gamma_{u,S,((i,C),H)}(J^\sharp,L^\sharp,W^\sharp,r^\sharp) = \{ t \mid t \in \T_S \cap \T_{(i,C)},\loc\,t = u, (J,L,W,r) = \beta\,t,\\
  \qquad \bigcup \{\gamma_{\I^\sharp} i \mid i\in J^\sharp\} \subseteq J,\\
  \qquad \forall a\in\M, \Cluster\in\Clusters_a: \exists\,v: (L\,(a,\Cluster) \cup \{\returnVar \mapsto v\}) \in \gamma_\R\,(L^\sharp\,(a,\Cluster)), \\
  \qquad W^\sharp \supseteq W, \exists\,v: (r' \cup \{\returnVar \mapsto v\}) \in \gamma_\R\,r^\sharp\\
  \quad \}
  \end{array}
\]
\[
  \begin{array}{lll}
  \gamma_{((i,C),H)}(J^\sharp,L^\sharp,r^\sharp) = \{ t \mid t \in \T_{(i,C)}, \last\,t = (\return\,x), (J,L,W,r) = \beta\,t,\\
  \qquad \bigcup \{\gamma_{\I^\sharp} i \mid i\in J^\sharp\} \subseteq J,\\
  \qquad \forall a\in\M: \exists\,v: (L\,a \cup \{\returnVar \mapsto v\}) \in \gamma_\R\,(L^\sharp\,a),\\
  \qquad \exists\,v: (r' \cup \{\returnVar \mapsto v\}) \in \gamma_\R\,r^\sharp\\
  \quad \}
\end{array}
\]
\[
  \begin{array}{lll}
  \gamma_{a,\Cluster,((i,C),H)}(r) =\{ t \mid t \in \T, (\last\,t = \unlock(a),\\
  \qquad \luW_a\,t=(\bar u_i,\_,\_),\\
  \qquad t' = (\bar u_i)\downarrow_t, t' \in \T_{(H\,a)_2}, (J,L,W,r) = \beta\,t',\\
  \qquad \exists\,v: (\beta\,t\cup \{\returnVar \mapsto v\}) \in \gamma_\R\,r\\
  \quad \} \cup \{t\in\T_{(H\,a)_2}\cap\init, (J,L,W,r) = \beta\,t, \exists\,v: (\beta\,t\cup \{\returnVar \mapsto v\}) \in \gamma_\R\,r\} \\
  \end{array}
\]
and remark that these concretization functions are monotonic.

Side-effecting in the abstract constraint system then can be abandoned whenever the ego thread definitely did not
write to any global from $\GM\,[a]$ since acquiring $a$.
The latter holds whenever $W^\sharp \cap \GM\,[a] = \emptyset$ at the given $\unlock(a)$.
By the same argument, whenever the ego thread has actually written a global from $\GM\,[a]$
since acquiring $a$, it's thread \emph{id} coincides with the thread \emph{id} of the thread
executing the last $\unlock(a)$ after a write to any global from $\GM[a]$.
The correctness proof of this construction follows along the same lines as the proofs of \cref{t:sound} and \cref{t:refined}.

In the last step,
control-point splitting at mutexes is then reduced to only consider the thread \emph{id} of the ego-thread
at the last $\unlock(a)$ immediately succeeding the last write to a global protected by $a$, i.e., $H\,a$.
Likewise, control-point splitting at program points and thread \emph{id}s is reduced back to the original $\A$ information proposed in \cref{s:unique}.
We remark that in this way, the infinite splitting of control-points disappears once again.
\qed

\section{One-element Clusters}\label{e:one}
  The following program illustrates that one-element clusters cannot be abandoned. Assume that $a$ protects both $g$ and $h$. 
   \begin{center}
      \begin{minipage}[t]{4cm}
    \begin{minted}{c}
    main:
     x = |\op{create}|(t1);
     y = |\op{create}|(t2);
     |\op{lock}|(a);
     h = 31;
     |\op{unlock}|(a);
     |\op{lock}|(a);
     h = 12;
     |\op{unlock}|(a);
     |\op{lock}|(a);
     // ASSERT(g<=h);  (1)
     // ASSERT(h==12); (2)
     |\op{unlock}|(a);
    \end{minted}
      \end{minipage}
      \begin{minipage}[t]{4.3cm}
    \begin{minted}{c}
    t1:
     |\op{lock}|(a);
     g =- 1;
     //ASSERT(g<=h);  (3)
     |\op{unlock}|(a);
     return 0;

    t2:
     |\op{lock}|(a);
     h = ?;
     h = 12;
     |\op{unlock}|(a);
     return 0;
    \end{minted}
      \end{minipage}
    \end{center}
  When running the clustered analysis with the cluster $\Clusters_a = \{\{g,h\}\}$ alone,
  the side-effect at the $\unlock(a)$ in $t_1$ preserves the relationship $g \leq h$,
  implying that the assertions (1) and (3) succeed.
  No precise information on the value of $h$ is preserved at the unknown $[a,\{g,h\},t_1]$.
  Consequently, when the \emph{main} thread performs a $\lock(a)$ for the third time,
  the assertion (2) cannot be verified.
  %
  A clustered analysis, though, that additionally tracks the cluster $\{h\}$,
  will record $h=12$ at $[a,\{h\},t_2]$ and nothing at $[a,\{h\},t_1]$.
  Therefore, assertion (2) can be verified.

\section{Extended Description of Benchmarks}
\label{a:more-experiments}
Here, we provide additional details on some of the benchmarks
presented in \cref{s:experiments}.

\subsection{\textsc{Goblint} Benchmark Suite}
We considered the benchmarks from~\cite{traces1}.
The benchmarks were augmented with asserts which were generated by our tool, using the \emph{Clusters} configuration.
We excluded \texttt{ypbind} from the experiments, as it spawns a thread from an unknown function pointer that our analysis of thread \emph{id}s cannot handle.
We executed our tool, as well as \textsc{NR-Goblint} and \textsc{Duet} on the resulting benchmarks.
Results are shown in \cref{tab:sas}.
For nine benchmarks, \textsc{NR-Goblint} and all configurations of our analyzer could prove all asserts, while our runs of \textsc{Duet}
did not produce valid results.
In the table, the results for these benchmarks are summarized in one row.
For the three remaining benchmark programs, using the \emph{Octagon} configuration instead of \emph{Interval} resulted in more proven asserts.
Interestingly, \textsc{NR-Goblint} could prove more asserts than our \emph{Interval} configuration.
Adding the analysis for thread \emph{ids} to our tool yielded more proven asserts on \texttt{pfscan}.
Indeed, for this set of benchmarks, all asserts can be proven with the \emph{TIDs} configuration.

\begin{table}[t]
    \centering
    \caption{Results of executing our analyzer (with presented analyses), \textsc{NR-Goblint}~\cite{traces1} (with intervals) and \textsc{Duet}~\cite{Farzan12} on the benchmark set by~\citet{traces1}, for which we automatically generated invariants using the analysis from \cref{s:clustering}.
        Checkmark~(\ding{51}) indicates all invariants proven and otherwise the number of proven invariants is given.
        For brevity, we summarize benchmarks that were fully proven by all analyzers except \textsc{Duet}, which reported on fewer invariants than present or timed out, so we consider those results invalid~(---).}
    \label{tab:sas}
    \newcommand{\s}{\cellcolor{green!30}\ding{51}}
\newcommand{\e}{\cellcolor{red!30}\ding{55}}
\newcommand{\w}[1]{\cellcolor{yellow!30}{#1}}
\newcommand{\invalid}{---}
\begin{tabular}{lccccc@{\hspace{1.5em}}cc}
    \toprule
    & & \multicolumn{4}{@{}c@{\hspace{1.5em}}}{Our analyzer} & \\
    \cmidrule(lr{1.6em}){3-6}
    Benchmark & Invariants & \parbox[c]{1.1cm}{\centering Interval \\ (Sec.~\labelcref{S:BASE})} & \parbox[c]{1.2cm}{\centering Octagon \\ (Sec.~\labelcref{S:BASE})} & \parbox[c]{1cm}{\centering TIDs \\ (Sec.~\labelcref{S:SELF-EXCLUDED})} & \parbox[c]{1.2cm}{\centering Clusters \\ (Sec.~\labelcref{s:clustering})} & \parbox[c]{2cm}{\centering \textsc{NR-Goblint} \\ w/ interval} & \textsc{Duet} \\
    \midrule
    pfscan & 221 & \w{14} & \w{32} & \s & \s & \w{165} & \invalid \\
    aget & 10 & \w{2} & \s & \s & \s & \s & \invalid \\
    ctrace & 1448 & \w{1130} & \s & \s & \s & \w{1407} & \invalid \\
    \addlinespace
    (Other) & 0--200 & \s & \s & \s & \s & \s & \invalid \\
    \bottomrule
\end{tabular}

\end{table}

\subsection{\textsc{Watts} Benchmark Suite}
While we were not able to run \textsc{Watts}, we executed our analysis on the benchmarks reported on in~\cite{Kusano16}.
We took the benchmarks as available from the GitHub repository\footnote{https://github.com/markus-kusano/watts}.
The benchmarks proposed in this paper consist of two set of C-source files containing multi-threaded code with asserts.
Unlike the other sets of benchmarks considered, there was a significant number of \texttt{assert(0)} statements in the benchmark files,
where code should be proven to be unreachable.
Consequently, for these benchmarks, we include the number of asserts that could be proven unreachable in the number of verified asserts we report.

The first set of benchmarks consists of 37 C-files, originating from different sources, that were adapted for~\cite{Kusano16}.
We took the benchmark set as-is, except for removing an obviously misplaced semicolon in the \texttt{wdt977\_02} benchmark
that rendered one \texttt{assert(0)} to be reachable in the concrete.
The analysis setting \emph{TIDs} succeeded to verify that none of the asserts in 31 of these benchmarks fail.
The six benchmark files for which in total seven asserts could not be proven, contain data-dependent thread-synchronization that our tool cannot handle.

The second set of benchmarks consists of five versions of the benchmark \texttt{i8xx\_tco\_03}, contained in the first benchmarks set, instrumented to create different numbers of threads.
In~\cite{Kusano16}, the number of threads created in the benchmarks varies from 30 to 70.
In the repository, there were two files creating 40 threads, and no file that created 30, rendering one test case redundant.
Thus, we removed ten thread creates from \texttt{i8xx\_tco\_03\_thr01}.
We also fixed the number of function stubs in that file to be the same as in the other benchmark files.

The runtimes of our tool can be seen in \cref{tab:watts-scale}.
Our most expensive analysis
takes around two seconds to complete on this benchmark.
While exact runtimes are not reported in~\cite{Kusano16}, the graph (found in Fig. 11 of that paper) indicates that the runtime of their
most expensive analysis was close to 400 seconds, while the least expensive configuration still took more than 200 seconds on the
benchmark creating 70 threads.
We remark that while runtimes reported for \textsc{Watts} in~\cite{Kusano16} were obtained on a different machine and the numbers are thus not directly comparable,
the comparison is still meaningful as the magnitudes differ greatly.

\begin{table}[t]
    \centering
    \caption{Runtimes, in seconds, of our analyzer on the five scalability benchmarks from~\cite{Kusano16}.
    The second column indicates the number of concrete threads in the benchmark program (including the \texttt{main} thread).
    Runtimes are considerably lower than those reported for \textsc{Watts}, but were obtained on different hardware.}
    \label{tab:watts-scale}
    \newcommand{\s}{\cellcolor{green!30}\ding{51}}
\newcommand{\e}{\cellcolor{red!30}\ding{55}}
\newcommand{\w}[1]{\cellcolor{yellow!30}{#1}}
\newcommand{\invalid}{---}
\begin{tabular}{lccccc@{\hspace{1.5em}}cc}
    \toprule
    & & \multicolumn{4}{@{}c@{\hspace{1.5em}}}{Our analyzer} & \\
    \cmidrule(lr{1.6em}){3-6}
    Benchmark  & \#Threads &\parbox[c]{1.1cm}{\centering Interval \\ (Sec.~\labelcref{S:BASE})} & \parbox[c]{1.2cm}{\centering Octagon \\ (Sec.~\labelcref{S:BASE})} & \parbox[c]{1cm}{\centering TIDs \\ (Sec.~\labelcref{S:SELF-EXCLUDED})} & \parbox[c]{1.2cm}{\centering Clusters \\ (Sec.~\labelcref{s:clustering})}  \\
    \midrule
    i8xx\_tco\_03\_thr01 & 31  & 0.15 & 0.16 & 0.43 & 0.51 &  \\ 
    i8xx\_tco\_03\_thr02 & 41 & 0.17 & 0.17 & 0.66 & 0.77 &  \\
    i8xx\_tco\_03\_thr03 & 51 & 0.19 & 0.19 & 1.01 & 1.11 &  \\
    i8xx\_tco\_03\_thr04 & 61 & 0.20 & 0.20 & 1.35 & 1.45 &  \\
    i8xx\_tco\_03\_thr05 & 71 & 0.23 & 0.22 & 1.83 & 1.99 &  \\
    \bottomrule
\end{tabular}

\end{table}

\subsection{Comparison with \textsc{Duet}}
We considered comparing to \textsc{Duet} on the benchmark set proposed in~\cite{Farzan12},
but encountered problems.
The archive of benchmarks was obtained from the \textsc{Duet} website\footnote{http://duet.cs.toronto.edu/}.
This archive contains (1)~the C source files of a set of Linux device drivers, without a harness function;
(2)~a version of these drivers as binary \texttt{goto}-files, that were instrumented with \textsc{DDVerify} and compiled with an unknown version of \textsc{Cbmc}.
However, neither the current version of \textsc{Duet} accepts \texttt{goto}-files as input,
nor any of the other tools considered.
We were also not able to decompile the \texttt{goto}-files using the current version of \textsc{Cbmc}.
Thus, these benchmarks were not used for experiments in the present paper.

While we managed to run \textsc{Duet} successfully on some benchmarks, our configuration of the tool did not produce valid results for others:
For dealing with code containing function calls, \textsc{Duet} relies on inlining.
As the inlining implemented in the most recent version of \textsc{Duet} available at the time of writing was not working on some examples,
we contacted the author, who sent us a fixed version of the module responsible for inlining.
In the experiments, we executed the tool with this implementation of inlining (and Octagons enabled);
however, there were still cases in which our configuration of the tool reported a too low number of reachable asserts, indicating
that some reachable code was not considered by the tool.
Thus, for these benchmarks, no results are reported for \textsc{Duet}.
Further, we do not report results for \textsc{Duet} when execution did not complete within 15 minutes.

\ignore{
Thus, we had to instrument the driver source files with \textsc{DDVerify} ourselves.

For dealing with function calls, \textsc{Duet} relies on inlining.
As the inlining implemented in the most recent version of \textsc{Duet} available was not working,
we contacted the author, who sent us a fixed version of the module responsible for inlining.
Still, this implementation of inlining was not capable of dealing with calls to (statically known) function pointers,
which occur when creating harnesses for the Linux device drivers with \textsc{DDVerify}.
We tried the current version of \textsc{Cbmc}, but this did not produce valid C-code on the instrumented device drivers.
Thus, we were not able to (meaningfully) redo the experiments presented in~\cite{Farzan12} with the current version of the tool.
}

\ignore{
\subsection{Our Benchmarks}
We give some further details on the benchmark results reported in \cref{s:experiments}.
\cref{tab:toy} shows the results of running our analyzer and \textsc{Duet} on a set of
our regression tests.
The results show that analysis precision profits from relational information, when comparing
the analysis using Octagons with Intervals.
Further improvements are obtained by TIDs and using clusters of size at most 2.
Multiple of the test cases cannot be handled by \text{Duet}

\todo{Describe toy benchmarks}
\begin{table}
    \centering
    \caption{Results of executing our analyzer (with presented analyses) and \textsc{Duet}~\cite{Farzan12} on our benchmark set.
        Checkmark~(\ding{51}) indicates all assertions proven, crossmark~(\ding{55}) indicates no assertions proven and anything in between is indicated by the number of proven assertions.
        In some cases \textsc{Duet} reported on fewer assertions than present, so we consider those results invalid~(---).}
    \label{tab:toy}
    \newcommand{\s}{\cellcolor{green!30}\ding{51}}
\newcommand{\e}{\cellcolor{red!30}\ding{55}}
\newcommand{\w}[1]{\cellcolor{yellow!30}{#1}}
\newcommand{\invalid}{---}
\begin{tabular}{lccccc@{\hspace{1.5em}}c}
    \toprule
     & & \multicolumn{4}{@{}c@{\hspace{1.5em}}}{Our analyzer} & \\
    \cmidrule(lr{1.6em}){3-6}
    Benchmark & Assertions & \parbox[c]{1.1cm}{\centering Interval \\ (Sec.~\labelcref{S:BASE})} & \parbox[c]{1.2cm}{\centering Octagon \\ (Sec.~\labelcref{S:BASE})} & \parbox[c]{1cm}{\centering TIDs \\ (Sec.~\labelcref{S:SELF-EXCLUDED})} & \parbox[c]{1.2cm}{\centering Clusters \\ (Sec.~\labelcref{s:clustering})} & \textsc{Duet} \\
    \midrule
    branched-not-too-brutal & 1 & \s & \s & \s & \s & \s \\
    tid-curious & 1 & \s & \s & \s & \s & \s \\
    escape-local-in-pthread-simple & 2 & \s & \s & \s & \s & \w{1} \\
    \addlinespace
    traces-max-simple & 1 & \e & \s & \s & \s & \e \\
    queuesize-const & 20 & \e & \s & \s & \s & \invalid \\
    airline & 1 & \e & \s & \s & \s & \s \\
    traces-min-rpb1 & 2 & \e & \s & \s & \s & \w{1} \\
    traces-min-rpb2 & 2 & \e & \s & \s & \s & \e \\
    traces-cluster-based & 4 & \e & \s & \s & \s & \e \\
    traces-write-centered-problem & 1 & \e & \s & \s & \s & \e \\
    mine14 & 1 & \e & \s & \s & \s & \e \\
    mine14-5b & 1 & \e & \s & \s & \s & \e \\
    traces-write-centered-vs-... & 2 & \e & \s & \s & \s & \e \\
    \addlinespace
    tid-toy3 & 2 & \e & \e & \s & \s & \w{1} \\
    pfscan-workers-strengthening & 1 & \e & \e & \s & \s & \s \\
    tid-toy5 & 2 & \e & \e & \s & \s & \w{1} \\
    tid-toy1 & 1 & \e & \e & \s & \s & \e \\
    tid-toy8 & 3 & \e & \e & \s & \s & \e \\
    tid-toy9 & 1 & \e & \e & \s & \s & \e \\
    sync & 2 & \e & \e & \s & \s & \e \\
    tid-toy10 & 1 & \e & \e & \s & \s & \e \\
    tid-toy11 & 1 & \e & \e & \s & \s & \e \\
    tid-toy6 & 2 & \e & \e & \s & \s & \invalid \\
    tid-toy7 & 2 & \e & \e & \s & \s & \invalid \\
    no-loc & 1 & \e & \e & \s & \s & \e \\
    \addlinespace
    traces-mutex-meet-cluster12 & 2 & \e & \e & \w{1} & \s & \e \\
    traces-mutex-meet-cluster2 & 1 & \e & \e & \e & \s & \s \\
    \bottomrule
\end{tabular}

\end{table}
}

\end{document}